%% file: main.tex
\documentclass{lmcs}
\pdfoutput=1
\usepackage[utf8]{inputenc}

\usepackage{lastpage}
\lmcsdoi{19}{2}{7}
\lmcsheading{}{\pageref{LastPage}}{}{}%
{Oct.~12,~2022}{Apr.~25,~2023}{}

\usepackage{amsmath}
\usepackage{amssymb}
\usepackage{amsthm}
\usepackage{mathtools}
\usepackage{tikz}
\usepackage{stackrel}
\usetikzlibrary{shapes.arrows,shapes.geometric,chains,decorations.pathreplacing,,decorations.text,positioning,arrows,automata,fadings,shadows}
\usepackage{hyperref}
\usepackage{cleveref}
\usepackage{dsfont}
\usepackage{cmll}
\usepackage{perfectcut}
\usepackage{multicol}
\usepackage{mathpartir}
\usepackage{stmaryrd}
\usepackage{proof}
\usepackage{xspace}
\usepackage{tcolorbox}
\usepackage[normalem]{ulem}
\usepackage{soul}
\usepackage{ mathrsfs }

\input{macros}
\input{rules}

\newcommand{\inj}{*}

\newcommand{\tvs}[1]{\tv{#1}_\s}
\newcommand{\tvrs}[1]{\tv{#1}_\rho^\s}
\newcommand{\tvS}[1]{\tv{#1}^\S}
 
\newcommand{\suc}{{\mathsf{s}}} 
\newcommand{\s}{{\mathfrak{s}}} 
\renewcommand{\S}{{\mathfrak{S}}}

\def\cf{{\em cf.}}
\def\ie{{\em i.e.}\xspace}
\def\etal{{\em et al.}\xspace}

\keywords{realizabiliy, nonstandard arithmetic, 
stateful computations, ultrafilters, glueing}
\begin{document}
\title{Stateful realizers for nonstandard analysis}

\thanks{The authors would like to thank Alexandre Miquel for suggesting 
several ideas at the root of this work and Valentin Blot and Mikhail Katz, 
as well as the anonymous reviewers of \cite{DM21}, 
for their accurate remarks and suggestions.
The first author was supported by FCT - Funda\c{c}\~ao para a Ci\^encia e 
a Tecnologia, and the research centers  CMAFcIO - Centro de Matem\'atica, 
Aplica\c{c}\~oes Fundamentais e Investiga\c{c}\~ao Operacional and  
CIMA - Centro de Investigação em Matemática e Aplicações -- 
under the projects UIDB/04561/2020, UIDP/04561/2020 and UIDP/04674/2020.
}

\author[B.~Dinis]{Bruno Dinis\lmcsorcid{0000-0003-2143-3289}}[a]
\author[É.~Miquey]{Étienne Miquey\lmcsorcid{0000-0002-5987-6547}}[b]
\address{Escola de Ciência e Tecnologia, Universidade de Évora}	
\email{bruno.dinis@uevora.pt}

\address{Aix-Marseille Université, CNRS, I2M, Marseille, France}	
\email{etienne.miquey@univ-amu.fr}  

\begin{abstract}
In this paper we propose a new approach to realizability interpretations for nonstandard arithmetic.
We deal with nonstandard analysis in the context of (semi) intuitionistic realizability, focusing on the Lightstone-Robinson construction of a model for nonstandard analysis through an ultrapower.
In particular, we consider an extension of the $\lambda$-calculus with a memory cell,
that contains an integer (the state), in order to indicate in which slice of the ultrapower 
$\M^{\N}$ the computation is being done. 
We pay attention to the nonstandard principles (and their computational content) obtainable in this setting. In particular, we give non-trivial realizers to Idealization and a non-standard version of the LLPO principle.
We then discuss how to quotient this product to mimic the Lightstone-Robinson construction.
\end{abstract}
\maketitle

\section{Introduction}
In this paper we propose a new approach to realizability interpretations for nonstandard arithmetic.
On the one hand, we deal with nonstandard analysis in the context of (semi) intuitionistic realizability.
On the other hand, we focus on Lightstone and Robinson's construction of a model for nonstandard analysis through an ultrapower \cite{LightstoneRobinson}. This paper is an extended version of \cite{DM21}. 
The main novelties here are in \Cref{s:ef}, where we establish a connection with evidenced frames~\cite{CohMiqTat21}, and in \Cref{s:LLPO}, where we give a realizer for a nonstandard version of the Lesser Limited Principle of Omniscience (LLPO). We also now have a better understanding of why performing a quotient leads to some counter-intuitive properties (\cf\, \Cref{s:some_name}).

Throughout the history of mathematics, infinitesimals were crucial for the intuitive development of mathematical knowledge by authors such as Archimedes, Stevin, Fermat, Leibniz, Euler and Cauchy, to name but a few  (see e.g.\ \cite{Katz(13), Blairetal(18),Bairetal(20)}).
 In particular, in Leibniz's Calculus one may recognize calculation rules -- sometimes called the \emph{Leibniz rules} \cite{Lutz(87),Callot(92),DinisBerg(19)} -- which correspond to heuristic intuitions for how the infinitesimals should operate under calculations: the sum and product of infinitesimals is infinitesimal, the product of a limited number (i.e. not infinitely large) with an infinitesimal is infinitesimal,...

In \cite{Robinson61,Robinson66} Robinson showed that, in the setting of model theory, it is possible to extend usual mathematical sets ($\mathbb{N}$, $\mathbb{R}$, etc.) witnessing the existence of new elements, the so-called \emph{nonstandard} individuals. In this way, it is possible 
to deal consistently with infinitesimal and
infinitely large numbers via ultraproducts and ultrapowers, in a way that is consistent with the Leibniz rules. Since the extended structures are nonstandard models of the original structures, this new setting was dubbed \emph{nonstandard analysis}. 

These constructions are meant to simplify doing mathematics: notions like limits or continuity can for instance be given a simpler form in nonstandard analysis.
Later in the 70s, Nelson developed a syntactical approach to nonstandard analysis, introducing in particular three key principles: Idealization, Standardization and Transfer~\cite{Nelson77}. 
The validity of these principles for constructive mathematics has been studied in different settings, in particular, following some pioneer work by Moerdijk, Palmgren and Avigad \cite{Moerdijk(95),MoerPalm(97),Avigad(01)} in nonstandard intuitionistic arithmetic, several recent works, inspired by Nelson's approach, lead to interpretations of nonstandard theories in intuitionistic realizability models~\cite{BerBriSaf12,DinFer16,HadVDB17,DinGas18}.

The very first ideas of \emph{realizability} are to be found in the 
Brouwer-Heyting-Kolmogorov interpretation~\cite{Heyting34,Kolmogoroff1932},
which identifies evidences and computing proofs (the realizers).
Realizability was designed by Kleene to interpret the computational 
content of the proofs of Heyting arithmetic~\cite{Kleene45},
and was later extended to more expressive frameworks \cite{Godel58, Kreisel51, Krivine09}.
While the Curry-Howard isomorphism focuses on a syntactical correspondence
between proofs and programs, 
realizability rather deals with the (operational) semantics of programs: 
a \emph{realizer} of a formula $A$ is a program which \emph{computes} adequately
with the specification that $A$ provides.
As such, realizability constitutes a technique to develop new models of 
a wide class of theories (from Heyting arithmetic to Zermelo-Fraenkel set theory),
whose algebraic structures has been studied in~\cite{VanOosten08,Krivine16,Miquel20}.

With the development of his classical realizability, 
Krivine evidences the fact that extending the $\lambda$-calculus 
with new programming instructions may result in getting new reasoning principles:
\texttt{call/cc} to get classical logic~\cite{Griffin90,Krivine09}, 
\texttt{quote} for dependent choice \cite{Krivine01}, etc.
In this paper, we follow this path to show how the addition of a monotonic
reference allows us to get a realizability interpretation for 
nonstandard analysis.
The realizability interpretation proposed here can be understood
as a computational interpretation of the ultraproduct construction in \cite{LightstoneRobinson}, 
where the value of the reference indicates the slice of the product in which
the computation takes place.
In particular, we obtain a realizer for the Idealization principle 
whose computational behaviour increases the reference in the manner 
of a diagonalization process. Our setting turns out to be semi-intuitionistic since it allows to deduce a rather non-trivial realizer for a nonstandard version of the nonconstructive Lesser Limited Principle of Omniscience \cite{BridgesRichman}.

\paragraph*{Outline}
We start this paper by recalling the main ideas of the  
ultraproduct construction (\Cref{s:robinson}) 
and the definition of a standard realizability interpretation 
for second-order Heyting arithmetic (\Cref{s:realizability}).
We then introduce stateful computations and our notion of 
realizability with slices in \Cref{s:slices}. We also show that our stateful interpretation induces an evidenced frame thus providing a connection with the usual algebraic tools to deal with realizability interpretations. 
As shown in \Cref{s:nonstandard} and \Cref{s:LLPO}, this interpretation provides us with
realizers for several nonstandard reasoning principles.
We discuss the possibility of taking a quotient
for this interpretation in \Cref{s:some_name}.
We conclude the paper in \Cref{s:conclusion} with a comparison to related work
and with some questions left for future work.

\section{The ultrapower construction}
\label{s:robinson} 
The main contribution of this paper consists in defining a realizability interpretation
to give a computational content to the ultrapower construction of Robinson and Lightstone in~\cite{LightstoneRobinson}.
We shall begin by briefly explaining how this construction works in the realm of model theory.

{Let us start by recalling some definitions.
\begin{defi}
 Let $I$ be a set. We say that $\F\subset\P(I)$ is a \emph{filter} over $I$ if:
 
 \begin{enumerate}[$(i)$]
  \item $\F$ is non empty and $\emptyset\notin\F$\hfill(\emph{non triviality})
  \item for all $F_1,F_2\in\F$,  $F_1\cap F_2\in\F$\hfill(\emph{closure under intersection})
  \item for any $F,G\in\P(I)$, if $F\in\F$ and $F\subset G$, then $G\in\F$\hfill(\emph{upwards closure})
 \end{enumerate}
 An \emph{ultrafilter} is a filter $\U$ such that for any $F\in\P(I)$, 
 either $F$ or its complement $\overline{F}$ are in $\U$.
\end{defi}
}
For instance, the set of cofinite subsets of $\N$ defines 
the so-called \emph{Fréchet filter}, which is not an ultrafilter since it contains
neither the set of even natural numbers nor the set of odd natural numbers.
Nonetheless, it is well-known that any filter $\F$ over an infinite set $I$ 
is contained in an ultrafilter $\U$ over $I$: this is the so-called \emph{ultrafilter principle}.
An ultrafilter that contains the Fréchet filter is called a \emph{free ultrafilter}.
The existence of free ultrafilters was proved by Tarski in 1930 \cite{Tarski1930} and
is in fact a consequence of the Axiom of Choice. 

{\begin{defi}
Given  two sets $V$ and $I$ and an ultrafilter $\U$ over $I$, 
we can define an equivalence relation $\cong_\U$ over $V^I$ by  $u \cong_\U\! v  \defeq \{i\in I:u_i=v_i\}\in\U$.
We write  $V^{I}/\U$ for the set obtained by performing a quotient on the set $V^I$ by this equivalence relation,
which is called an \emph{ultrapower}.
\end{defi}
}

Consider a theory $\T$ (say ZFC) and its language $\mathcal{L}$,
for which we assume the existence of a model $\M$.
The goal is to build a nonstandard model $\M^*$ 
of the theory $\T$ that validates new principles.
Let us denote by $\V$ the set which interprets individuals in $\M$,
and let us fix a free ultrafilter $\U$ over $\N$. 
Roughly speaking, the new model $\M^*$ is defined as the ultrapower $\M^\N/\U$.
Individuals are interpreted by functions in $\V^\N$
while the validity of a relation $R(x_1,...,x_k)$ 
(where the $x_i$ are interpreted by $f_i$, for $i\in\{1,...,k\}$) is defined by
\[ \M^* \vDash R(f_1,...,f_k) \qquad\text{ iff }\qquad \{n\in\N:\M\vDash R(f_1(n),...,f_k(n))\} \in \U. \]
We can now extend the language with a new predicate $\st{x}$ to express that $x$ is \emph{standard}. 
Standard elements are defined as the ones that, with respect to $\cong_\U$, 
are equivalent to constant functions, i.e. 
$\M^* \vDash \st f$ if and only if there exists $p\in\N$ such that $\{n\in\N: f(n)=p\} \in \U$.
{\[ \M^* \vDash \st f \qquad\text{ iff }\qquad \exists p\in\N.\{n\in\N: f(n)=p\} \in \U. \]}
Formulas that involve this new predicate are called \emph{external},
while formulas of the original language $\mathcal{L}$ are called \emph{internal}.

Lightstone and Robinson's construction relies on the well-known \Los' theorem~\cite{Los60}
which states that if $\varphi$ is an internal formula (with parameters in $\V^\N$),
then $\M^*\vDash \varphi$ if and only $\{n\in\N:\M\vDash \overline{\varphi}^n\}\in\U$,
where $\overline{\varphi}^n$ refers to the formula $\varphi$ whose parameters have been 
replaced by their values in $n$. 
This construction indeed defines a model of $\T$ which satisfies other relevant properties,
namely Transfer, Idealization and Standardization.
As a consequence of \Los' theorem, to see that an internal formula $\varphi(x)$ 
holds for all elements, it is enough to see that it holds for all standard elements:
this is the \emph{Transfer} principle.
In our setting, \emph{Idealization} amounts to a diagonalization process: it is for instance easy to see
that if one defines $\delta:n\mapsto n$ (where we, with abuse of notation, write $n$ for both the natural number $n$ 
and its interpretation in $\V$), then $\M^*\vDash \forall x.(\stto{x} x<\delta)$.
Finally, \emph{Standardization} is a sort of ``comprehension scheme'' which
states that we can specify subsets of standard sets by giving a membership criterion 
for standard elements (by means of an internal formula).

\section{Realizability in a nutshell}
\label{s:realizability}
\subsection{Heyting second-order arithmetic}
\label{s:int_real}
We start by introducing the terms and formulas of Heyting second-order arithmetic (HA2), for which
we follow Miquel's presentation~\cite{Miquel11}. 
Second-order formulas are build on top of first-order arithmetical expressions, 
by means of logical connectives, first- and second-order quantifications 
and primitive predicates.
We use upper case letters for second-order variables and lower case letters for first-order ones.
We use a primitive predicate $\nat{e}$
to denote that $e$ is a natural number ($0$ then has type $\nat{0}$ and the term $\suc\,t$ has type
$\nat{S(e)}$ provided that $t$ has type $\nat{e}$). 
We consider the usual $\lambda$-calculus terms extended with pairs, 
projections (written $\pi_i$),
natural numbers and a recursion operator:
\[\begin{array}{lr@{~~}c@{~~}l}
\text{\bf 1\textup{st}-order expressions} & e &::= & x \mid 0 \mid S(e) \mid f(e_1,\dots,e_n)\\
\text{\bf Formulas} & A,B   &::= & \nat{e} \mid X(e_1,\dots,e_n)\mid A\imp B \mid A\land B\\ 
                    &       &  & \mid  \forall x.A \mid \exists x.A \mid \forall X.A \mid \exists X.A  \\
\text{\bf Terms   } & t,u   &::= & x \mid 0 \mid \suc \mid \rec \mid \lambda x.t\mid t\,u \mid (t,u) \mid \pi_1 (t) \mid \pi_2 (t) \\
\end{array}\]
where $f:\N^n \to \N$ is any arithmetical function. 
We write $\Lambda$ for the set of all closed $\lambda$-terms.

Note that we did not include disjunctions as formulae. 
In fact, for most of our purposes, disjunctions are not needed. 
The only exception is \Cref{s:LLPO} where we interpret a nonstandard version of 
the Lesser Limited Principle of Omniscience. 
It is however possible to include disjunctions all the way but, 
since this would mostly only add some unnecessary technicalities to all our proofs,
we delay the introduction of disjunction until \Cref{s:LLPO}.

To simplify the use of existential quantifiers, as in~\cite{Miquel11},
we introduce a congruence relation on formulas defined by the following rules
\begin{equation}
(\exists x. A) \imp B \cong \forall x.( A \imp B)
\qquad\qquad\qquad
(\exists X. A) \imp B \cong \forall X.( A \imp B)
\label{eq:congruences}
\end{equation}
This congruence relation allows us to, given any typed term, to (re)type it with any 
formula congruent to the original one. 
In particular, this means that we do not need the elimination rules for the existential 
quantifiers, which results in a simplified type system. 
This type system, which is given in \autoref{fig:HA2_types},
corresponds to the usual rules of natural deduction.
The reader may observe that we do not give computational content to quantifications.

In the sequel, we make use of the following usual abbreviations:
\[
\begin{array}{rcl}
\suc^{n+1} 0   &\defeq& \suc\,(\suc^n 0)            \\
\overline{n}   &\defeq& \suc^n 0                    \\
\end{array}
\quad\vrule \quad 
\begin{array}{rcl}
\top     &\defeq& \exists X.X                       \\
\bot     &\defeq& \forall X.X                       \\
\neg A   &\defeq& A \to \bot \\
\end{array}
\quad\vrule\quad 
\begin{array}{rcl}
e = e'  &\defeq& \forall Z. (Z(e) \imp Z(e'))      \\
\far x.A &\defeq& \forall x.(\nat{x} \imp A)        \\
\exr x.A &\defeq& \exists x.(\nat{x} \land A)       \\
\end{array}
\]
It is well-known that the above definition of equality 
(often called \emph{Leibniz law}) enjoys the usual expected properties 
(reflexivity, symmetry, transitivity) and allows to perform substitution of equal terms.
The quantifications $\far x.A$ and $\exr x.A$ are often said to be \emph{relativized}
to natural numbers.

\renewcommand{\TYP}[3]{#1\vdash_{\mathrm{NJ}} #2:#3}
\renewcommand{\TYP}[3]{#1\vdash #2:#3}
\newcommand{\dblinefill}{
\leavevmode \cleaders \hbox to.5em{\hss=\hss}\hfill\kern0pt
}
\begin{figure}[t]
    {{
    \hrulefill
    \vspace{.8em}
    \begin{flushleft}
      \textbf{Natural numbers}   
    \end{flushleft}
 \vspace{-3em}

    \begin{mathpar}
    
    \hspace{2cm}
    
    \infer[\autorule{0}]{\TYP{\Gamma}{0}{\nat{0}}}{}
    
    \infer[\autorule{S}]{\TYP{\Gamma}{\suc}{\far x.\nat{S(x)}}}{}
    
    \infer[\recrule]{\Gamma\vdash \rec:\forall Z. Z(0) \imp (\far y.(Z(y)\imp Z(S(y)))) \imp \far x.Z(x)}{}
    \end{mathpar}

    \hrulefill
    \vspace{0.3em}

    \begin{flushleft}
    \textbf{Logical rules}
    \end{flushleft}
     \vspace{-3em}
    
    \begin{mathpar}
    
      \hspace{0.5cm}
      
    \infer[\axrule]{\TYP{\Gamma}{x}{A}}{(x:A)\in\Gamma}
    
    \infer[\imperule]{\TYP{\Gamma}{t\,u}{B}}{
          \TYP{\Gamma}{t}{A\to B} &\quad \TYP{\Gamma}{u}{A}
        }
        
    \infer[\impirule]{\TYP{\Gamma}{\Lam{x}{t}}{A\to B}}{
         \TYP{\Gamma,x:A}{t}{B}
    }

    \infer[\andirule] {\TYP{\Gamma}{(t,u)   }{A \land B}}{\TYP{\Gamma}{t}{A} & \TYP{\Gamma}{u}{B}} 
    
    \infer[\andeurule]{\TYP{\Gamma}{\pi_1(t)}{ A        }}{\TYP{\Gamma}{t}{A \land B}}
    
    \infer[\andedrule]{\TYP{\Gamma}{\pi_2(t)}{ B        }}{\TYP{\Gamma}{t}{A \land B}}

    \infer[\exiurule]       {\TYP{\Gamma}{ t}{\exists x.A           }}{\TYP{\Gamma}{t}{A[x:=n]}    }
    
    \infer[\feurule]        {\TYP{\Gamma}{ t}{ A[x:=n]              }}{\TYP{\Gamma}{t}{\forall x.A}}
    
    \infer[\fiurule]        {\TYP{\Gamma}{ t}{\forall x.A           }}{\TYP{\Gamma}{t}{A} & x\notin \FV(\Gamma)}
    
    \infer[\exidrule]       {\TYP{\Gamma}{ t}{\exists X.A           }}{\TYP{\Gamma}{t}{A[X(x_1,\dots,x_n):=B]}}
    
    \infer[\fedrule]        {\TYP{\Gamma}{ t}{A[X(x_1,\dots,x_n):=B]}}{\TYP{\Gamma}{t}{\forall X.A}}
    
    \infer[\fidrule]        {\TYP{\Gamma}{ t}{\forall X.A           }}{\TYP{\Gamma}{t}{ A  }& X\notin \FV(\Gamma)}
    
    \infer[\autorule{\cong}]{\TYP{\Gamma}{ t}{A                     }}{\TYP{\Gamma}{t}{ A' }& A\cong A'}
    
    \end{mathpar}
    \hrulefill
    }}
\caption{Type system}     \label{fig:HA2_types}
\end{figure}
\renewcommand{\red}{\rightarrowtail}
\renewcommand{\red}{\rightarrow_\beta}
\newcommand{\stepb}[2]{#1 \triangleright_\beta #2}
\newcommand{\axinf}[2]{\overline{\stepb{#1}{#2}}}
\newcommand{\unaryinf}[4]{\infer{\stepb{#1}{#2}}{\stepb{#3}{#4}}}
\newcommand{\binaryinf}[6]{\infer{\stepb{#1}{#2}}{\stepb{#3}{#4} & \stepb{#5}{#6}}}

The one-step (weak) reduction over terms is defined by the following rules:

\begin{mathpar}
  
\axinf{(\lambda x.t)u}{t[u/x]}

\axinf{\rec\,u_0\,u_1\,0}{u_0}

\axinf{\rec\,u_0\,u_1\,(\suc\,t)}{u_1\,t\,(\rec\,u_0\,u_1\,t)}

\axinf{\pi_1(t,u)}{t}

\axinf{\pi_2(t,u)}{u}

\end{mathpar}
We write $\red$ for the congruent reflexive-transitive closure of $\stepb{}{}$. 
The reduction $\red$ is known to be confluent, type-preserving and normalizing on typed terms~\cite{Barendregt93}.

\subsection{Realizability interpretation of HA2}\label{s:real}
 In this subsection we define the realizability interpretation of the type system defined in \autoref{fig:HA2_types}, 
 in which formulas are interpreted as \emph{saturated sets of terms}, i.e.\ as sets of closed terms 
 $S\subseteq \Lambda$ such that ${t}\red{t'}$ and $t'\in S$ imply that $t\in S$. 
 We write {\sat} to denote the set of all saturated sets and, given a formula $A$, we
 call \emph{truth value} its realizability interpretation.

 \renewcommand{\gets}{\mapsto}
 \begin{defi}[Valuation]
   A \emph{valuation} is a function~$\rho$ that associates a
  natural number $\rho(x)$ to every first-order variable~$x$ and
  a \emph{truth value function} $\rho(X)$, i.e.\ a function in $\N^k\to\sat$ to every
  second-order variable~$X$ of arity~$k$.
  \begin{enumerate}
  \item Given a valuation~$\rho$, a first-order variable~$x$ and a
    natural number $n$, we denote by $\rho,x\gets n$ the
    valuation defined by $(\rho,x\gets n)~\defeq~\rho_{|\dom(\rho)\setminus\{x\}}\cup\{x\gets n\}\,$.
  \item Given a valuation~$\rho$, a second-order variable~$X$ of
    arity~$k$ and a truth value function $F:\N^k\to\sat$, the valuation defined by $(\rho,X\gets F)~\defeq~
    \rho_{|\dom(\rho)\setminus\{X\}}\cup\{X\gets F\}\,$
    will be denoted by $\rho,X\gets F$.
  \end{enumerate}
  We say that a valuation $\rho$ is \emph{closing} the formula $A$ if $\FV(A)\subseteq \dom(\rho)$.
  
\end{defi}

 \begin{defi}[Realizability interpretation]\label{def:realizability}
 We interpret closed arithmetical expressions $e$ in the standard model of first-order Peano arithmetic $\N$. Given a valuation $\rho$ and a first-order expression $e$ (whose variables are in the domain of $\rho$) we denote its interpretation by $\foint{e}_{\rho}$.
 The interpretation of a formula $A$ together with a valuation $\rho$ closing $A$
 is the set $\tvr{A}$ defined inductively according to the following clauses:
\[
\begin{array}{rcl}
 \tvr{\nat{e}         }& \defeq & \{t \in\Lambda : t \red \suc^n 0 \text{, where } n = \foint{e}_{\rho}\} \\[2mm]
 \tvr{X(e_1,\dots,e_n)}& \defeq & \rho(X)(\foint{e_1}_{\rho},\dots,\foint{e_n}_{\rho}) \\[2mm] 
 \tvr{A\imp B         }& \defeq & \{t\in\Lambda : \forall u\in\tvr{A}.(t\,u \in\tvr{B})\} \\[2mm] 
 \tvr{A_1\land A_2    }& \defeq & \{t\in\Lambda : \pi_1(t)\in\tvr{A_1} \land \pi_2(t)\in\tvr{A_2}\}  \\[2mm]

 \tvr{\forall x.A     }& \defeq & \bigcap_{n\in\N}\tv{A}_{\rho,x\gets n} 
 \qquad\qquad  \\[2mm] 
 \tvr{\exists x.A     }& \defeq & \bigcup_{n\in\N}\tv{A}_{\rho,x\gets n} \\[2mm]
 \tvr{\forall X.A     }&\defeq  &\bigcap_{F:\N^k\to\sat}\tv{A}_{\rho,X\gets F}\\[2mm]
 \tvr{\exists X.A     }&\defeq & \bigcup_{F:\N^k\to\sat}\tv{A}_{\rho,X\gets F}
 \end{array}
\]
\end{defi}
Observe that in the previous definition, the universal quantifications
cannot be seen as generalized conjunctions. 
Indeed, the conjunction is given computational content through pairs, while the universal quantifications are defined as intersections of truth values.

It is easy to see that for any formula $A$ and any valuation $\rho$ 
closing $A$, one has $\tvr{A}\in\sat$.
As it turns out, the congruences defined by \autoref{eq:congruences} are sound w.r.t.~the interpretation.
\begin{propC}[\cite{Miquel11}]\label{r:cong}
 If $A$ and $A'$ are two formulas of HA2 such that $A \cong A'$, then for all valuations $\rho$ 
closing both $A$ and $A'$ we have $\tvr{A} = \tvr{A'}$.
\end{propC}
\begin{proof}
 By induction on $A\cong A'$. Congruence easily goes through by induction, 
 we only prove the first-order case (the second-order case is analogous):
\[\begin{array}{rl}
\tvr{(\exists x. A) \imp B}
& = \{t\in\Lambda : \forall u \in \tvr{\exists x. A}, t\,u \in \tvr{B}\} \\
& = \{t\in\Lambda : \forall u \in \bigcup_{n\in\N}\tv{A}_{\rho,x\gets n}, t\,u \in \tvr{B}\} \\
& = \{t\in\Lambda : \forall u, (\exists n,u\in\tv{A}_{\rho,x\gets n}) \imp t\,u \in \tvr{B}\} \\
& = \{t\in\Lambda : \forall u, \forall n, (u\in\tv{A}_{\rho,x\gets n} \imp t\,u \in \tvr{B})\} \\
& = \bigcap_{n\in\N}\{t : \forall u, u \in\tv{A}_{\rho,x\gets n} \imp t\,u \in\tvr{B}\} \\
& = \tvr{\forall x.( A \imp B)}
\end{array}\vspace{-2em}\] 
\end{proof}

In order to show that the realizability interpretation is sound with respect to
the type system we need the following preliminary notions.

\begin{defi}[Substitution]
  A \emph{substitution} is a finite function~$\sigma$ from
  $\lambda$-variables to closed $\lambda$-terms.
  Given a substitution~$\sigma$, a $\lambda$-variable~$x$ and a closed
  $\lambda$-term~$u$, we denote by $(\sigma,x:=u)$ the substitution
  defined by
  $(\sigma,x:=u)\defeq\sigma_{|\dom(\sigma)\setminus\{x\}}\cup\{x:=u\}$.
\end{defi}

\begin{defi}\label{def:subst}
Given a context $\Gamma$ and a valuation $\rho$ closing the formulas in $\Gamma$, 
we say that a substitution~$\sigma$ \emph{realizes}~$\rho(\Gamma)$ and 
write $\sigma \real \rho(\Gamma)$ if $\dom(\Gamma)\subseteq \dom(\sigma)$ and
$\sigma(x)\in \tvr{A}$ ~for every declaration $(x:A)\in\Gamma$.
\end{defi}

\begin{defi}\label{def:adequate}
 A typing judgement $\TYP{\Gamma}{t}{A}$ is \emph{adequate}
 if for all valuations $\rho$ closing $A$ and $\Gamma$
 and for all
  substitutions $\sigma\real\rho(\Gamma)$ we have
  $\sigma(t)\in \tvr{A}$. 
More generally, we say that an inference rule
  $\infer{J_0}{J_1 &\cdots& J_n}$
  is adequate  if the adequacy of all typing
  judgements $J_1,\ldots,J_n$ implies the adequacy of the typing
  judgement $J_0$.
\end{defi}

\begin{thm}[Adequacy \cite{Miquel11}]\label{r:Adequacy}
  The typing rules of \autoref{fig:HA2_types} are adequate.
\end{thm}
\begin{proof}
The proof is standard, by case analysis.
We draw the reader's attention to the particular case of the second-order
elimination rule
\[ \infer[\fedrule]{\TYP{\Gamma}{ t}{A[X(x_1,\dots,x_n):=B]}}{\TYP{\Gamma}{t}{\forall X.A}} \]
which relies on the fact that the truth value of any formula (here $B$) is a saturated set.
To prove that this rule is indeed adequate, 
let us consider a valuation $\rho$ closing $\forall X.A$, $B$ and $\Gamma$ and a substitution
$\sigma\real\rho(\Gamma)$ such that $\sigma(t)\in\tvr{\forall X.A}$.
By definition, this implies that for any function $F:\N^k \to \sat$ (where $k$ is the arity of $X$),
we have $\sigma(t)\in\tv{A}_{\rho,X\mapsto F}$.
To conclude, it suffices to see that the function $n_1,\dots,n_n\mapsto \tv{B}_{\rho,x_1\mapsto n_1,\ldots,x_k\mapsto n_k}$ is indeed in $\N^k \to \sat$. \qedhere

\end{proof}

\begin{cor}
 If $\Gamma \vdash t:A$ is derivable, then it is adequate.
\end{cor}

The adequacy theorem is the key result when defining realizability 
interpretations in that fundamental properties stem from it. For example, we have the following corollary.
\begin{cor}[Consistency]\label{r:consistency} 
 There is no proof term $t$ such that $\vdash t:\bot$.
\end{cor}
\begin{proof}
 The proof is by \emph{reductio ad absurdum}. If $\vdash t:\bot$, then by \autoref{r:Adequacy} one has $t\in\tv{\bot}=\tv{\forall X.X}=\bigcap_{S\in\sat} S = \emptyset$. To see that this intersection is indeed empty, one can take for example $S_0 = \{t\in\Lambda : t\red 0\} \in \sat$ and $S_1 = \{t\in\Lambda : t\red\suc 0\} \in \sat$, then clearly $S_0 \cap S_1 = \emptyset$.
\end{proof}

We would like to point out that the proof of adequacy is very flexible. Indeed, if one wants to add a new instruction to the language of terms via its typing rule, it is enough to check that this typing rule is adequate while the remainder of the proof is exactly the same.

\subsection{Introducing value restriction}
The realizability interpretation of \autoref{def:realizability} is also flexible 
regarding the set of formulas that are interpreted. 
We illustrate this point here by introducing a new construction extending formulas which
we shall use in the sequel to enforce value restriction in presence of stateful computations.
This will allow us to get a better handle on the operational semantics,
which will be crucial  afterwards since stateful computations break the confluence 
of the reduction system (see \autoref{ex:confluence}). 
Such a technique is reminiscent from ML value restriction that was introduced
to circumvent the incompatibility of Curry-style polymorphism and side-effects 
(see for instance~\cite{Zeilberger09}).

We start by defining the subset $\V\subseteq \Lambda$ of \emph{values} by the following grammar: 
\[
\text{\bf Values  } \qquad\qquad V  \quad::= \quad 0 \mid \suc\, V \mid \lambda x.t\mid (V_1,V_2) \qquad
\]

Observe that variables are not values, otherwise the system would not be stable by substitution.
In the remainder of this paper, we adopt the convention that $\lambda$-terms
are denoted by lowercase letters $t,u,...$ while 
uppercase letters $V,W,...$ refer to values.

Distinguishing the set of values allows for instance to
restrict the $\beta$-reduction rule to applications of functions to values:
\newcommand{\redv}{\triangleright_\byv}
\begin{mathpar}
\infer{(\lambda x.t)V \redv t[V/x]}{}

\infer{t\,u \redv t'\,u}{t \redv t'}

\infer{V\,u\redv V\,u'}{u \redv u'}
\end{mathpar}
The reflexive transitive closure $\to_\byv$ of the one-step reduction $\triangleright_\byv$ is
known as the (left-to-right) \emph{call-by-value} evaluation strategy. 
While it is well-known that the reduction system of the $\lambda$-calculus is confluent, 
so that the choice of a particular evaluation strategy does not have any consequence
in terms of expressiveness, this is no longer the case when 
side effects (such as stateful computations in the next sections)
come into play.

{
To enforce value restriction, let us now extend the language of formulas with a new construction:
\[
\text{\bf Formulas} \qquad\qquad A,B   \quad::= \quad \ldots \mid \bvto A B \qquad
\]
and the realizability interpretation accordingly by
\[ \tvr{\bvto A B  }\quad \defeq \quad \{t\in\Lambda : \forall V\in\tvr{A}.(t\,V \in\tvr{B})\} \]
In particular, we have
\[
\begin{array}{rcl}
\tvr{\bvto {\nat{e}} B  }& = & \{t\in\Lambda : t\,\overline{n}\in\tvr{B} \text{ where } n = \foint{e}_{\rho} \} \\
\tvr{\bvto {A_1\land A_2} B}& = & \{t\in\Lambda :\forall V_1\in\tvr{A_1},V_2\in\tvr{A_2}. t\,(V_1,V_2)\in\tvr{B} \text{ where } n = \foint{e}_{\rho} \} \\
\end{array}
\]
}
It is easy to check that for any formulas $A$ and $B$, 
$\tvr{\bvto A B  }$ is a saturated set, and the adequacy of the \fedrule-rule is thus preserved.

While there is currently no rule to type a term $t$ with a formula of the shape 
$\bvto A B$, we can nonetheless extend the type system with any rule as long as it is adequate with respect
to the realizability interpretation. 
Indeed, here the flexibility of the interpretation comes again into play, in the sense
that once the realizability interpretation of a new construct has been defined,
one could extend the type system with any rule related with that construct
as long as it is adequate.
For instance, the rules $\autorule{\mapsto_\intromark}$ and $\autorule{\mapsto_\elimmark}$ below
are adequate.
\begin{prop}\label{r:bvto}
 The following typing rules are adequate:
 \twoelt{
 \infer[\autorule{\mapsto_\intromark}]{\Gamma\vdash t:\bvto A B}{\Gamma\vdash t:A \imp B}
 }{
 \infer[\autorule{\mapsto_\elimmark}]{\Gamma\vdash t\,V:B}{\Gamma\vdash t:\bvto A  B & \Gamma \vdash V:A}
 }
 \end{prop}
\begin{proof}
 For the first rule it suffices to see that for any valuation $\rho$, we have
 \[ \{t\in\Lambda : \forall u\in\tvr{A}.(t\,u \in\tvr{B})\}\subseteq \{t\in\Lambda : \forall V\in\tvr{A}.(t\,V \in\tvr{B})\} \]
 As for the second one,
 if $\rho$ is a valuation and $\sigma$ a substitution such that $\sigma(V)\in\tvr{A}$,  $\sigma \Vdash \rho(\Gamma)$, and
 $\sigma(t)\in\tvr{\bvto A B}$.
 Then, by the definition of $\tvr{\bvto A B}$, we have that $\sigma(t)\,\sigma(V)=\sigma(t\,V)\in\tvr{B}$ (because $\sigma(V)$ is necessarily a value).
 \end{proof}

  We can also extend, maintaining the adequacy of the interpretation of $\bvto A B$ , the congruence relation with the following rules:
 \twoelt{
  \bvto {\exists x. A} B \cong {\forall x.\bvto { A} B}
  }{
  \bvto {\exists X. A} B \cong {\forall X.\bvto { A} B}
  }

 \begin{prop}\label{r:bv_cong} 
 For any formulas $A$ and $B$, we have
 
   \begin{enumerate}
  \item $\tvr{\bvto {\exists x. A} B} = \tvr{\forall x.\bvto { A} B}$
  \item $\tvr{\bvto {\exists X. A} B} = \tvr{\forall X.\bvto { A} B}$
 \end{enumerate}
\end{prop}

{
\begin{proof}
The proof is analogous to the proof of \autoref{r:cong}, for instance for the first part, we have:
\[\begin{array}{rl}
\tvr{\bvto {\exists x. A}  B}
& = \{t\in\Lambda : \forall V \in \tvr{\exists x. A}, t\,V \in \tvr{B}\} \\
& = \{t\in\Lambda : \forall V \in \bigcup_{n\in\N}\tv{A}_{\rho,x\gets n}, t\,V \in \tvr{B}\} \\
& = \{t\in\Lambda : \forall V, (\exists n,V\in\tv{A}_{\rho,x\gets n}) \limp t\,V \in \tvr{B}\} \\
& = \{t\in\Lambda : \forall V, \forall n, (V\in\tv{A}_{\rho,x\gets n} \limp t\,V \in \tvr{B})\} \\
& = \bigcap_{n\in\N}\{t : \forall V,  \in\tv{A}_{\rho,x\gets n} \limp t\,V \in\tvr{B}\} \\
& = \tvr{\forall x.(\bvto A B)}
\end{array}\vspace{-2em}\]
\end{proof}
}

We will make use of the following abbreviations:
\twoelt{
 \fabv{x}.A  \defeq \forall x.(\bvto{\nat{x}} A) 
 }{
 \exbv{x}.A \defeq  \forall X.(\fabv x.(A  \imp X))\imp X
 }
 While the first definition is natural, 
the second one may be a bit more puzzling at first sight.
As we saw, the truth value of any formula has to be a saturated set.
However, given a formula $A(x)$, the set $\{(\overline{n},t):t\in\tvr{A(n)}\}$ 
is not saturated, and so we cannot define a formula 
$\exists x.\{\nat{x}\}\land A(x)$ whose realizers would be this set.
Nonetheless, the definition of $\exbv{x}.A$ is somehow doing the trick
in continuation-passing style, in the sense that we have:
\begin{prop}\label{r:exbv}
 For any formula $A$, any valuation $\rho$ and any term $t$,
 if $t\in\tvr{\exbv{x}.A}$ then there exists a natural number $n\in\N$
 and a term $u\in\tvr{A[x:=n]}$ s.t.: $t\,(\lambda xy.(x,y))\red (\overline{n},u)$. 
\end{prop}
\begin{proof}
 Let $t$ be a term such that $t\in\tvr{\exbv{x}.A}$.
 By definition, 
 for any $\X\in\sat$ and any $v\in \tv{\fabv x.(A  \imp X)}_{\rho,X\mapsto\X}$, we have that
 $t\,v\in\X$.
 Let us define the set $$\X = \{w\in\Lambda:\exists n,u. w\red (\overline{n},u) \land u\in\tvr{A[x:=n]}\},$$
 which is obviously saturated.
 It is clear that $\lambda xy.(x,y)\in\tv{\fabv x.(A  \imp X)}_{\rho,X\mapsto\X}$
 since for any $n\in\N$ and any $u\in\tvr{A[x:=n]}$ one has
 $(\lambda xy.(x,y))\,\overline{n}\,u \red (\overline n,u)\in\X$.
 We conclude that 
 $t\,(\lambda xy.(x,y)) \red (\overline n,u)$.
 \end{proof}

\begin{defi}\label{d:T}
 We define $T\defeq \lambda zx.(\rec\,(\lambda y.y\,0)\,(\lambda xyz.y\,(\lambda x.z\,(\suc\,x)))\,x)\,z$. 
\end{defi}

The next proposition relates these new quantifications
 with  the relativized quantifications $\far x.A$ and $\exr x.A$
 using the term $T$.

\begin{prop}\label{r:relativized}
We have

 \begin{enumerate}
 \item  $T \Vdash \fabv x.A \to \far x.A$
 \item  $\lambda x.x \Vdash \far x.A \to \fabv x.A$
 \item  $\lambda z.z\,\lambda xy.(x,y) \Vdash \exbv x.A \to \exr x.A$
 \item  $\lambda xy.T\,y\,\pi_1(x)\,\pi_2(x) \Vdash \exr x.A \to \exbv x.A$
  \end{enumerate}
\end{prop}
{
\begin{samepage}
\begin{proof} \hfill
\begin{enumerate}
 \item 
Let $t$ be a term in $\tv{\fabv x.A}$,
$n\in\N$ a natural number
and $u$ a term in $\tv{\nat{n}}$.
To prove the result, since $\tv{A[x:=n]}$ is saturated, it suffices to prove that:
\[ (\rec\,(\lambda y.y\,0)\,(\lambda xyz.z\,(\suc\,x))\,u)\,t \Vdash A[x:=n] \]
Let us define the formula $B(z)\defeq (\fabv x.A) \to A[x:=z]$.
It is straightforward to check that:
\begin{itemize}
\item $\lambda y.y\,0\Vdash B(0)$
\item $\lambda xyz.y\,(\lambda x.z\,(\suc\,x)) \Vdash \far{x}. B(x)\to B(S(x))$
\end{itemize}
Using the adequacy of the \recrule-rule, we deduce that
\[ \rec\,(\lambda y.y\,0)\,(\lambda xyz.z\,(\suc\,x))\,u\Vdash B(n) \]
and the result follows from the hypothesis that $t\in\tv{\fabv x.A}$.
\item Follows directly from \autoref{r:bvto}.
\item Follows directly from \autoref{r:exbv}.
\item Let $t$ be a term in $\tv{\exr x.A}$, $\X\in\sat$ be a predicate and $u$ a term in $\tv{\fabv x.(A\to X)}_{X\mapsto\X}$.
By assumption, there exists a natural number $n$ and two terms $t_1$ and $t_2$ such that $t_1\in\tv{\nat{n}}$, $t_2\in\tv{A[x:=n]}$ and 
$t\red (t_1,t_2)$. 
Then
\[ (\lambda xy.T\,y\,\pi_1(x)\,\pi_2(x))\,t\,u 
~~\red~~ 
T\,u\,\pi_1(t)\,\pi_2(t) 
~~\red~~
T\,u\,t_1\,t_2 \]
From part 1 we know that $T\,u$ is in $\tv{\far x.(A\to X)}_{X\mapsto\X}$, hence
$T\,u\,t_1\,t_2\in\X$ and the result follows from the fact that $\X$ is saturated.\qedhere
\end{enumerate}
\end{proof}
\end{samepage}
}
The term $T$, which forces the evaluation of an argument of type $\nat{n}$ 
to get the underlying value $\overline{n}$ to make it compatible
with a function $\fabv x.A$, is somehow simulating a call-by-value
evaluation (for natural numbers). 
Such a term is usually called a \emph{storage operator}~\cite{Krivine09}.
 
While \autoref{r:relativized} indicates that the different ways of relativizing 
the quantifiers are equivalent (in the sense that one admits a realizer 
if and only if the other does), it is important to keep in mind that 
this result is peculiar to the current effect-free settings.
In particular, this result no longer holds once stateful computations are allowed.

\section{Realizability with slices}
\label{s:slices} 

In this section, we extend the realizability interpretation introduced in
\Cref{s:realizability} by taking into account stateful computations.
The states will be the key ingredient to give computational content
to the Lightstone-Robinson construction described in \Cref{s:robinson}.
The advantage of this extended setting will become
clearer in \Cref{s:nonstandard}, since it allows to realize some new reasoning principles,
after we investigate the status of natural numbers in \Cref{s:natural}.

\subsection{Stateful computations}
The first step in the Lightstone-Robinson construction aims at getting a product $\M^\N$ of the (initial) model $\M$.
In order to achieve this goal in our setting,  we add a memory cell to our calculus that contains an integer, 
which we call the \emph{state}.
The purpose of the state is to keep track of which ``slice'' of the product is the interpretation being done.
This product allows us to interpret first-order individuals as functions in $\N^\N$, so that
the interpretation accounts for new elements -- the so-called nonstandard elements -- for instance 
the diagonal function (see \autoref{r:delta}).

In our extended calculus, the first-order expressions are the same, while second-order formulas now use a value restriction for natural numbers and include a predicate $\st{e}$, as per usual in nonstandard analysis, 
denoting that the expression $e$ is standard. This means that in our framework we will also have two types of nonstandard quantifications: the usual $\fas{}, \exs{}$ and the relativized  $\fabvs, \exbvs{x}$.
We say that a formula is \emph{internal} if it does not contain the predicate $\st{\cdot}$, and
\emph{external} otherwise.
Terms are extended with two new instructions $\get$ and $\set$.
Similarly to what was remarked before concerning the $\bvto{\cdot}\cdot $
construction on formulas for value restriction,
we do not need a specific rule for typing these terms since
any adequate rules would work. 
Instead, we will only pay attention to the computational expressiveness
brought by these new terms, by introducing appropriate reduction rules.
The $\get$ instruction allows to obtain the content of the current state
while the $\set$ instruction allows to increase its content. 
Formally, we extend the different grammars as follows:
\[\begin{array}{l@{\qquad}rcl}

\text{\bf Formulas} & A,B  &::= & \st{e}  \mid X(e_1,\dots,e_n) \mid \bvto{\nat e} A \mid  A\imp B \\
&&&\mid A\land B 
\mid  \forall x.A \mid \exists x.A \mid \forall X.A \mid \exists X.A  \\
\text{\bf Terms   } & t,u   &::= & ... \mid \get \mid \set  \\
\text{\bf States}   &\S &\defeq & \N
\end{array}
\]
Since the formulas no longer include an unrestricted constructor $\nat e$, the typing rules for $0$, $\suc$ and 
$\rec$ are no longer required\footnote{In \autoref{r:natural}, we show how these terms define realizers for the value restricted natural numbers.}. 
Other than that, the type system is unchanged.
In particular, the $\get$ and $\set$ instructions are not given any typing rule.

We will make use of the following abbreviations:
\[
\begin{array}{rcl}

 \fas{x}.A   &\defeq& \forall x.(\stto{x} A) \\

 \fabvs{x}.A &\defeq& \forall x.(\stto{x} (\bvto{\nat x} A)) \\

 \end{array} \quad\vrule~
 \begin{array}{rcl}
 \exs{x}.A   &\defeq&  \exists x.(\st{x} \land A) \\

 \exbvs{x}.A &\defeq& \forall X.((\fabvs x.(A\to X)) \to X) \\

 \end{array}
\]

\newcommand{\reds}[2]{\xrightarrow{\!\!#1\rightsquigarrow #2\!\!}}
\renewcommand{\reds}[2]{\,{}^{#1}\!\!\downarrow^{#2}}

\newcommand{\step}[2]{#1 \triangleright^\s_{\s} #2}
\newcommand{\steps}[4]{#3 \triangleright^{#1}_{#2} #4}

With the exception of the $\get/\set$ instructions, 
the syntax of terms does not account for states. In fact, only the reduction rule for the $\set$ instruction allows to change the state.
Nonetheless, states play a crucial role in the reduction system. 
In particular, one-step reductions are now defined for terms together with a state.
We write $\steps{\s}{\s'}{t}{t'}$ to denote that the term $t$ in state $\s$ reduces to
the term $t'$ in state $\s'$.
The one-step reduction over terms is defined by the following rules:
\newcommand{\hsep}{\qquad\qquad}
\[
\newcommand{\unaryinfs}[4]{\infer{\steps{\s}{\s'}{#1}{#2}}{\steps{\s}{\s'}{#3}{#4}}}
\begin{array}{c}
\infer{\steps \s \s t t'}{t \triangleright_\beta t'}
\hsep
\infer{\steps \s \s \get {~\s}}{}
\hsep
\infer{\steps{\s'}{\s''}{\set\overline {\s} \,t}{t}}{\s''=\max(\s,\s')}
\hsep
\unaryinfs{C[t]}{C[t']}{t}{t'}
\hsep
\end{array}
\]
\newcommand{\hole}{[\,]}
where $C\hole::= \rec\,u_0\,u_1\,\hole 
~\mid~ \hole\,u 
~\mid~ \pi_i(\hole)
~\mid~ \suc\,\hole
~\mid~ \set\,\hole\,u$.

We write $t\reds{\s}{\s'} t'$ for the reflexive-transitive closure of this relation.

Since we now consider effectful computations, we have to fix an evaluation strategy in order to ensure the confluence
of the reduction system.
Observe that our definition for $C\hole$ indeed ensures that
our reduction system has no critical pair. 
Here we follow a call-by-name evaluation strategy (we substitute unevaluated arguments), 
while for $\rec$ and $\set$ one of their arguments must be reduced.

The following standard example illustrates the need for an evaluation strategy to ensure
confluence in the presence of states, highlighting the fact that the result of a stateful 
computation might depend on the chosen strategy.
\begin{exa}\label{ex:confluence} 
Let us write $x + y$ for a term that computes the addition of $x$ and $y$ (such term is easily definable via $\rec$).
Let us define $\incr_0\defeq \set\,(\suc\,\get)\,0$ (which increases the state and reduces to $0$) 
and $t\defeq(\lambda x.(\get+x) + x)\,\incr_0$.
If we reduce the argument of the functions first (call-by-value) we obtain
$t\reds{0}{1}(\lambda x.(\get+x)+x))\,0\reds{1}{1} (\get+0)+0 \reds{1}{1} 1$. 
In turn, if we perform the $\beta$-reduction without reducing the argument (call-by-name), 
we get $t\reds{0}{0} (\get+\incr_0) + \incr_0 \reds{0}{1} (\get+\incr_0)+0 \reds{1}{2} \get+0\reds{2}{2} 2$.
In the absence of an evaluation strategy, the system would thus have admitted unsolvable critical pairs.
\end{exa}

\subsection{Stateful realizability interpretation}
\label{s:stateful_real}
The fact that our syntax now includes states allows us
to interpret formulas as terms-with-states\footnote{A realizability interpretation
with a similar structure, although with a different notion of state, can be found in~\cite{MiqHer18}. 
The perspective of the latter is also different in that it aims at proving 
the normalization of a classical call-by-need calculus.}.
Truth values are then defined as saturated sets in $\P(\Lambda \times \S)$.
Individuals are now individuals with states, so elements of $\N^\S$,
and similarly predicates of arity $k$ are elements of the set of functions from $\N^k$ to $\P(\Lambda \times \S)$.
This creates a mismatch in the sense that predicates are no longer shaped to be
applied to individuals\footnote{This phenomenon also occurs in the Lightstone-Robinson
construction of an ultrapower~\cite{LightstoneRobinson}.}.
In order to define our interpretation, we need to deal with
this mismatch between the structure of individuals and the one of predicates,
by defining a suitable notion of application.

\begin{defi}\label{def:application}
 Let $F:\N^k\to\P(\Lambda\times\S)$ be a predicate.
We define the application of $F$ to individuals $f_1,\ldots,f_k\in\N^\S$ by 
$F@(f_1,\ldots,f_k) \defeq \{(t;\s): (t;\s)\in F(f_1(\s),\ldots,f_k(\s))\}$.
\end{defi}

\begin{defi}\label{def:standard}
An individual $f\in\N^\S$ is said to be \emph{standard} if it is a constant function,
i.e.\ if there exists $n\in\N$ such that $\forall \s\in\S.(f(\s)=n)$. We then write $f=n^\inj$.

\end{defi}

\begin{defi}
We define \emph{saturated sets with respect to the stateful reduction}
to be sets $S\in\Lambda\times\S$
s.t. for any terms $t,t'\in\Lambda$ and any states $\s,\s'\in\S$,
if $(t';\s')\in S$ and $t \reds \s {\s'} {t'}$ then $(t;\s)\in S$. 
With abuse of notation we denote the set of these saturated sets by $\sat$.
\end{defi}
In the realizability interpretation with slices below, truth values are defined as saturated sets.
This allows us to reason by \emph{anti-reduction} (sometimes also called \emph{expansion}) in any fixed state. 
By anti-reduction, we mean that to show that a term 
$t$ together with a state $\s$ belongs to such a saturated set $S$, it is enough to find
$\s'$ and $t'$ such that $t\reds{\s}{\s'}t'$ and $(t';\s')\in S$.

We now consider valuations which are functions that associate 
  a function in $\N^\S$ to every first-order variable~$x$ and
  a truth value function from $\N^k$ to $\sat$ to every
  second-order variable~$X$ of arity~$k$.
Again, with abuse of notation we denote such valuation by $\rho$.

We also extend the usual interpretation of first-order expressions 
to range over $\N^\S$. To that end, we simply define arithmetical functions
pointwise on the domain. For instance, if $f\in\N^\S$,
we write $S^\inj(f)$ for the function $\s \mapsto (S(f(\s)))$.
When it is clear from the context, we abuse the notation (even more) by writing
$0$, $S$, $\foint{\cdot}_{\rho}$, etc. instead of $0^\inj$, $S^\inj$, $\foint{\cdot}^\inj_{\rho}$.

\renewcommand{\tvr}[1]{\tvS{#1}_{\rho}}
\newcommand{\prim}{'}

\begin{defi}[Realizability with slices]
 \label{def:real_slices}
The interpretation of a formula $A$ together with a valuation $\rho$ closing $A$ 
is the set $\tvr{A}$ defined inductively according to the following clauses:
 \[
\begin{array}{rcl}

 \tvr{\st{e}}          & \defeq & \left\{\begin{array}{ll} \Lambda\times\S & \text{ if $\foint{e}_{\rho}$ is standard}\\ \emptyset & \text{otherwise}\end{array}\right.\\[2mm]

 \tvr{X(e_1,\dots,e_n)}& \defeq & \rho(X)@(\foint{e_1}_{\rho},\dots,\foint{e_n}_{\rho}) \\[2mm] 

 \tvr{\bvto {\nat e} A  }     & \defeq & \{(t;\s)\in\Lambda\times\S : (t\,\overline{n};\s)  \in\tvr{A} \text{, where } n = \foint{e}_{\rho}(\s)\} \\[2mm]
 \tvr{A\imp B         }& \defeq & \{(t;\s)\in\Lambda\times\S : \forall u. \big((u;\s)\in\tvr{A} \limp (t\,u;\s) \in\tvr{B}\big)\} \\[2mm]
 
\tvr{A_1\land A_2        }& \defeq & \{(t;\s)\in\Lambda\times\S : (\pi_1(t);\s)\in\tvr{A_1} \land (\pi_2(t);\s)\in\tvr{A_2} \big)\}  \\[2mm]

 \tvr{\forall x.A}& \defeq & \bigcap_{f\in\N^\S}\tvS{A}_{\rho,x\gets f} \\[2mm]
 \tvr{\forall X.A} &\defeq&  \bigcap_{F:\N^k\to\sat}\tvS{A}_{\rho,X\gets F} \\[2mm] 
 \tvr{\exists x.A}& \defeq & \bigcup_{f\in\N^\S}\tvS{A}_{\rho,x\gets f} \\[2mm]
  \tvr{\exists X.A}& \defeq&  \bigcup_{F:\N^k\to\sat}\tvS{A}_{\rho,X\gets F}.
 \end{array}
\] 
We write $(t;s)\real A$ (resp. $t\ureal A$) to denote that $(t;s)\in\tvS{A}$ (resp. $\forall \s\in\S.(t;\s)\in\tvS{A}$). 
Realizers of the type $t\ureal A$ are called \emph{universal}.
\end{defi}
Observe that this stateful interpretation has the structure of a
product of the interpretation given by \autoref{def:realizability}. 
The interpretation corresponding to a given state can thus 
be seen as a \emph{slice} of this product. 
However, it is important to keep in mind that the $\set$ instruction still allows
terms to change the value of the state, 
therefore the slices are not completely independent.
We write $\tvrs{A}$ to denote the truth value $\{(t;\s)\in\tvr{A}\}$ in the slice induced by $\s$.

We first verify that truth values are indeed saturated sets and that the interpretation
validates the congruence rules.
\begin{prop}\label{r:real_sat}
 Let $A$ be a formula and $\rho$ a valuation closing $A$. Then $\tvr{A}\in\sat$.
\end{prop}
\begin{proof}
 \input{real_sat}
\end{proof}

\begin{prop}\label{r:cong_states}
 If $A$ and $A'$ are two formulas of HA2 such that $A \cong A'$, then for all valuations $\rho$ 
closing both $A$ and $A'$ we have $\tvr{A} = \tvr{A'}$.
\end{prop}
\begin{proof} 
The proof, by induction on $A\cong A'$, is similar to the proof of \autoref{r:cong}.
Congruence easily goes through by induction, and again we have
\[\begin{array}{rl}
\tvr{(\exists x. A) \imp B}
& = \{(t;\s)\in\Lambda\times\S: \forall u. (u;\s) \in \tvr{\exists x. A} \limp                      (t\,u;\s) \in \tvr{B}\} \\
& = \{(t;\s)\in\Lambda\times\S: \forall u. (u;\s) \in \bigcup_{n\in\N}\tv{A}_{\rho,x\gets n}\limp  (t\,u;\s) \in \tvr{B}\} \\
& = \{(t;\s)\in\Lambda\times\S: \forall u. (\exists n.(u;\s)\in\tv{A}_{\rho,x\gets n}) \limp   (t\,u;\s) \in \tvr{B}\} \\
& = \{(t;\s)\in\Lambda\times\S: \forall u, n. ((u;\s)\in\tv{A}_{\rho,x\gets n} \limp   (t\,u;\s) \in \tvr{B})\} \\
& = \bigcap_{n\in\N}\{(t;\s) : \forall u. (u;\s) \in\tv{A}_{\rho,x\gets n} \limp               (t\,u;\s) \in \tvr{B}\} \\
& = \tvr{\forall x.( A \imp B).} 
\end{array}\]
The proofs for second-order quantifiers and value restrictions are analogous.
\end{proof}

In order to prove the adequacy theorem in this setting we need to adapt a few definitions.

\begin{defi}\label{def:adequate_state}
Given a context $\Gamma$, a state $\s$ and a valuation $\rho$ closing the formulas in $\Gamma$, 
we say that a substitution~$\sigma$ \emph{realizes}~$\rho(\Gamma)$ in the state $\s$ and 
write $(\sigma;\s) \real \rho(\Gamma)$ if $\dom(\rho(\Gamma))\subseteq \dom(\sigma)$ and
$(\sigma(x);\s)\in \tvr{A}$, for every declaration $(x:A)\in\Gamma$.
\end{defi}

\begin{defi}
 We say that a typing judgement $\TYP{\Gamma}{t}{A}$ is \emph{adequate} w.r.t.\ a state $\s$
 in the stateful system  if for any valuation $\rho$ closing $A$ and $\Gamma$ and any 
  substitution $(\sigma;\s)\real\rho(\Gamma)$ we have $(\sigma(t);\s) \in \tv{\rho(A)}$. 
  An inference rule is adequate w.r.t.\ a state $\s$
  if the adequacy (w.r.t.\ $\s$) of all its premises
  implies the adequacy (w.r.t.\ $\s$) of its conclusion.
\end{defi}

We are now able to show that, with the exception of the {\fedrule}/{\exidrule}-rules, 
our logical rules are adequate. 
The {\fedrule}/{\exidrule}-rules are shown to be adequate, for internal formulas only, in \autoref{r:elimination}.
The status of natural numbers will be investigated in \Cref{s:natural}.

\begin{thm}[Adequacy]\label{r:state_adequacy}
  The logical rules of \autoref{fig:HA2_types}, except the {\fedrule}/{\exidrule}-rules, are adequate.
\end{thm}
\begin{proof}
\input{state_adequacy}

\end{proof}

\begin{rem}\label{rmk:pred}

  Let us explain why the \fedrule-rule is not adequate in general (the same argument
  applies to the {\exidrule}-rule). 
  As emphasized at the beginning of this section, we interpret predicates by functions
  from $\N^k$ to ${\sat}$,
  while the truth values of formulas may vary in the set of functions from $(\N^\S)^k$ to $\sat$. 
    \autoref{r:glueing} will make this more precise: internal formulas correspond to functions
  from $\N^k$ to ${\sat}$ while external formulas correspond to functions 
  from $(\N^\S)^k$ to ${\sat}$. 
  Therefore, in general we cannot substitute a second-order variable by any formula.
  Indeed, in the second-order elimination rule (for universal quantifiers) variables can only 
  be instantiated by internal formulas.
  Moreover, if the formula $B$ that we want to substitute is a proposition 
  (i.e.\ if its arity $k$ is equal to $0$), then the substitution is valid since the
  interpretations of internal and external formulas coincide.   
 This means that we could have chosen to work with impredicative encodings of the 
 conjunction (or other connectives) as in the Russell-Prawitz translation~\cite{Prawitz65}.
 Indeed, such an encoding relies on the use
of propositions, which are thus compatible with the elimination rule:

\[A \land B ~\defeq~ \forall X.(A \imp B \imp X) \imp X. \]
\end{rem}
\begin{rem}\label{rmk:elim}
We would like to attract the reader's attention to the fact that our realizability interpretation
is grounded in the elimination rules for the connectives. 
While this choice may not be so meaningful in a pure intuitionistic setting,
here the fact that our realizers may perform some effectful computations
makes this choice relevant.
Indeed, the other possibility would have been to require from a realizer
of $A\land B$ to be a term reducing to a pair of realizers, forcing
the effectful computations to be done  ``right away'', which could in particular
make the value of the state evolve. 
In turn, our definition delays such computations further, which allows us to reason within
the same state before eventually reducing the term (and thus performing the effect). 
This technicality turns out to be crucial in some proofs in the sequel, in particular
for defining the realizer of LLPO$^{\mathrm{st}}$ in \Cref{s:LLPO}.
\end{rem}

We show that $\rec$ realizes a formula that emulates its former typing rule by using quantifiers relativized with a value restriction.
\begin{prop}\label{r:rec_bv} 
We have $\rec\ureal {\forall X.X(0)\to \fabv{x}.(X(x) \to X(S(x))) \to \fabv{x}.X(x)}$.

\end{prop}
\begin{proof}
Let $\X:\N\to\sat$ be a predicate, $\s\in\S$ be a state, $f\in\N^\S$ be a natural number, 
 $u_0$ and $u_S$ be terms and $V$ be a value such that
 \begin{itemize}
  \item[-] $(u_0;\s)\in \X(0)$, 
  \item[-] $(u_S;\s)\in \tvS{\far y.(X(y)\imp X(S(y)))}_{X\mapsto \X}$

 \item[-] $(V;\s)\in\tvS{\nat{f}}$.
 \end{itemize}
The latter implies that $V=\suc^n 0$ where $n=f(\s)$.
Besides, recall that by definition we have $\tvS{X(f)}_{X\mapsto \X}=\X@(f)=\{(t;\s)\in \X(f(\s))\}=\X(n)$.
Let us prove, by induction on $n$, that  $\rec\,u_0\,u_S\,\overline{n}\in \X(n)$.
\begin{itemize}
\item If $n=0$, then we have that $\rec\,u_0\,u_S\,t \reds{s}{s}\rec\,u_0\,u_S\,\overline 0 \reds{s}{s} u_0$, the result
follows by anti-reduction from the hypothesis on $u_0$.
\item If $n=S(m)$, then we have that $\rec\,u_0\,u_S\,(\suc\, \overline m)  
\reds{s}{s} u_S\,\overline m\,(\rec\,u_0\,u_S\,\overline m)$.
By induction hypothesis, we have that $(\rec\,u_0\,u_S\,\overline m;\s)\in \X(m)$.
The result thus follows (by anti-reduction) from the hypothesis on $u_S$.
\qedhere
\end{itemize}
\end{proof}

\begin{rem}Regarding the necessity of restricting the relativization of quantifiers to values,
the proof of \autoref{r:rec_bv} is enlightening.
{Indeed, if instead of a value $V$ we were only given a term in $\tvrs{\nat{f}}$,
by definition this term may change the value of the state, say to some $\s'$,
before reducing to the value of $f(\s')$. 
This would break the proof since nothing is assumed on the realizers
$u_0$ and $u_S$ in this new state $\s'$.}
\end{rem}

\subsection{Glueing}
\newcommand{\trunc}[1]{\overline{#1}^\s}

An important property of our interpretation (which also reflects
a similar property in the Lightstone-Robinson construction) is that the interpretation
of internal formulas can be decomposed as the product of its slices (\autoref{r:glueing}).
In other words, internal formulas can only access information in the 
current state. In particular, and as expected, this means that it is
impossible to express standardness by means of internal formulas. 
To state this formally, we first define the restriction of formulas 
and truth values with respect to a slice.
\begin{defi}
Given an internal formula $A$, we define $\trunc{A}$ as the formula whose
individuals are all applied in $\s$.
Formally, it amounts to replacing each individual by the standard individual 
with which it coincides in the state $\s$:

 \[\begin{array}{r@{~}c@{~}l}
  \trunc{{X}(e_1 , . . . , e_k )} & \defeq & X((e_1(\s))^\inj, \ldots, (e_k(\s))^\inj ) \\[2mm]
  \trunc{A\imp B } & \defeq & \trunc{A}\imp\trunc{B} \\[2mm]
  \trunc{\bvnat e B } & \defeq & \bvnat {(e(\s))^\inj} {\trunc{B}} \\
  \end{array}
~~\vrule~~
 \begin{array}{r@{~}c@{~}l}
  \trunc{A \land B } & \defeq & {\trunc{A}} \land {\trunc{B}} \\[2mm]

  \trunc{\forall x.A }    &\defeq & \forall x .\trunc{A} \\
  \end{array}
~~\vrule~~
 \begin{array}{r@{~}c@{~}l}
   \trunc{\exists x.A }    &\defeq & \exists x .\trunc{A} \\[2mm]
  \trunc{\forall X.A }    &\defeq & \forall X .\trunc{A} \\[2mm]
  \trunc{\exists X.A }    &\defeq & \exists X .\trunc{A} \\
   \end{array}\]
 \end{defi}

\begin{thm}[Glueing]\label{r:glueing}
 For any internal formula $A$ and valuation $\rho$ closing $A$,
 we have that $(t;\s)\in\tvr{A} \Leftrightarrow t\in\tvrs{\trunc{A}}$.
\end{thm}
\begin{proof}
\input{glueing}
 \end{proof}

Let $B(x)$ be a formula whose only free variable is $x$, and $\rho$ a valuation closing $B$.
In general, the function $\F_B$ that associates to any individual $f$
the truth value $\tvr{B(f)}$ is a function from $\N^\S$ to $\sat$.
If $B$ is internal, by the glueing theorem, 
to determine $\F_B$ it is enough to know its value for standard individuals. 
This means that we only need to know a function from $\N$ to $\sat$.
As such, we can now formally state the intuition developed in \autoref{rmk:pred}.
\begin{prop}\label{r:elimination} 
The elimination  rule for the 2{\footnotesize{nd}}-order universal quantification and the  introduction rule for the 2{\footnotesize{nd}}-order existential quantification

\twoelt{
\infer[\fedrule]{\Gamma\vdash t : A[X(x_1,...,x_k):=B]}{\Gamma \vdash t:\forall X.A}}{
\infer[\exidrule]{\TYP{\Gamma}{ t}{\exists X.A       }}{\TYP{\Gamma}{t}{A[X(x_1,\dots,x_n):=B]}}
}
are adequate for any internal formula $B$ whose only free variables are $(x_1,...,x_k)$.

\end{prop}
\begin{proof}
\input{elimination}
\end{proof}

\begin{rem}
 Observe that external formulas such as $\stto{x} \bot$  cannot be defined by glueing.
 Consider for instance a nonstandard element $\tau$.
 Then $\tvS{\stto{\tau} \bot}=\Lambda\times\S$, while for any state $\s\in\S$
 we have $\tvs{\trunc{\stto{\tau} \bot}}=\tvs{{\stto{\tau(\s)^\inj} \bot}}=\tvs{\top\to\bot}=\emptyset$.
 \end{rem}
 
It is well-known that the comprehension scheme $\CA_B~\defeq~ \exists X.\forall x.(X(x) \Leftrightarrow B)$
is a logical consequence of the elimination principle
$\Elim^B_A ~\defeq ~ (\forall X.A) \limp A[X(x):=B]$ (by taking $A=\exists Y.\forall x.(Y(x) \Leftrightarrow X(x))$).
Since we have the \fedrule-rule restricted to internal formulas $B$,
the comprehension scheme is also valid for these formulas. 
In particular, this implies Standardization for internal formulas,
i.e.\ for $B$ an internal formula, the following holds

\[\fas X.\exs Y.\fas z.(Y(z) \Leftrightarrow X(z) \land B(z)).\]

Of course, the comprehension scheme does not hold for external formulas, 
so the relativization on the quantifiers in Standardization is in this sense necessary.
We will come back to Standardization in \Cref{s:quotient}.

\subsection{The induced evidenced frame}
\label{s:ef}

Before studying the properties of this interpretation,
we shall connect it with the usual algebraic tools to deal 
with realizability interpretation, in order to better emphasize its structure and 
peculiarities. In recent work, Cohen \etal have been introducing a new framework 
to capture the algebraic structure of realizability interpretations, 
which they named \emph{evidenced frames}~\cite{CohMiqTat21}. 
These have the benefit of being generic enough to easily
encompass effectful interpretation, 
while uniformly inducing triposes (and thus toposes), 
hence a model of higher-order logic.
We show here how our interpretation fits the picture, hinting in particular
at the possibility to extend our interpretation to
deal with higher-order logic (which is out of the scope of this paper, as here we want to focus on the second-order fragment only).

We start by recalling the definition of evidenced frame.
\begin{defiC}[\cite{CohMiqTat21}]\label{def:ef}
An \emph{evidenced frame} is a triple $( \Phi, E,  \mbox{$\cdot \xle{\cdot} \cdot$} )$, where $\Phi$ is a set of propositions,~$E$ is a collection of evidence, and~\mbox{$\phi_1 \xle{e} \phi_2$} is a ternary evidence relation on $\Phi \times E \times \Phi$, along with the following:
\begin{samepage} 
\begin{description}[leftmargin=10pt,font= {\itshape \mdseries}]
\item[Reflexivity] There exists evidence~$\eid \in E$:
\begin{itemize}[leftmargin=*]
\item $\forall \phi.\; \phi \xle{\eid} \phi$
\end{itemize}
\item[Transitivity] There exists an operator~$\ecomp{\cdot}{\cdot} \in E \times E \to E$:
\begin{itemize}[leftmargin=*]
\item $\forall \phi_1, \phi_2, \phi_3, e, e'.\; \mbox{$\phi_1 \xle{e} \phi_2 \mathrel{\wedge} \phi_2 \xle{e'} \phi_3 \implies \phi_1 \xle{\ecomp{e}{e'}} \phi_3$}$
\end{itemize}
\item[Top] A proposition~$\top \in \Phi$ such that there exists evidence~$\etrue \in E$:
\begin{itemize}[leftmargin=*]
\item $\forall \phi.\; \phi \xle{\etrue} \top$
\end{itemize}
\item[Conjunction] An operator~\mbox{$\wedge \in \Phi \times \Phi \to \Phi$} such that there
exists evidence~$\efst, \esnd \in E$
and 
an operator~\mbox{$\epair{\cdot}{\!\cdot} \!\in\! E \!\times\! E \!\to\! E$}:\\
\begin{tabular}{@{}l@{~~\qquad}l}
$\bullet$ $\forall \phi_1, \phi_2.\; \phi_1 \wedge \phi_2 \xle{\efst} \phi_1$
&
$\bullet$ $\forall \phi, \phi_1, \phi_2, e_1, e_2.\; \mbox{$\phi \xle{e_1} \phi_1 \mathrel{\wedge} \phi \xle{e_2} \phi_2 \!\!\implies\!\! \phi \xle{\epair{e_1}{e_2}} \phi_1 \wedge \phi_2$}$\\
$\bullet$ $\forall \phi_1, \phi_2.\; \phi_1 \wedge \phi_2 \xle{\esnd} \phi_2$
\end{tabular}

\item[Universal Implication] An operator~$\efimp \in \Phi \times \power(\Phi) \to \Phi$ such that there exists an operator~\hbox{$\elambda{} \in E \to E$} and evidence~$\eeval \in E$:
\begin{itemize}[leftmargin=*]
\item $\forall \phi_1, \phi_2, \vec{\phi}, e.\; (\forall \phi \in \vec{\phi}.\; \phi_1 \wedge \phi_2 \xle{e} \phi) \implies \phi_1 \xle{\elambda{e}} \phi_2 \efimp \vec{\phi}$
\item $\forall \phi_1, \vec{\phi}, \phi \in \vec{\phi}.\; (\phi_1 \efimp \vec{\phi}) \wedge \phi_1 \xle{\eeval} \phi$
\end{itemize}
\end{description}
\end{samepage}
\end{defiC}

The definition of the evidenced frame induced by our stateful interpretation
better highlights its core structure.
First, as shown by the interpretation of second-order variables,
propositions are defined as truth values, that is as saturated sets
of terms-with-states.
Evidences, in turn, are defined as universal realizers, {\ie} $\lambda$-terms, 
with the corresponding evidence relation 
\[S_1\overset{t}\to S_2~~\defeq~~ t\ureal \forall (u;\s)\in S_1. (t\,u;\s)\in S_2
\eqno (S_1,S_2\in\sat, t\in\Lambda)
\]

\begin{prop}\label{p:evidenced}
 The tuple $(\sat,\Lambda,\cdot\overset{\cdot}\to\cdot$) defines
 an evidenced frame.
\end{prop}
\begin{proof}
We give the evidence and constructors on propositions for each case. 
We mostly follow the realizability interpretation given in \autoref{def:real_slices}
\begin{description}[leftmargin=10pt,font= {\itshape \mdseries}]
\item[{Reflexivity}] It is clear that $\eid \defeq \lambda x.x \ureal S \to S$
for any $S\in\sat$.
\item[Transitivity]
For any $S_1, S_2, S_3\in\sat$ if $t_1\ureal S_1 \to S_2$ and $t_2\ureal S_2 \to S_3$, it is clear that $t_1;t_2 \defeq \lambda x.t_2\, (t_1\,x) \ureal S_1 \to S_3$.
\item[Top] Let $\top \defeq \Lambda\times\S\in\sat$. Then 
we have $\etrue \defeq \lambda x.x\ureal S\to\top$ for any $S\in\sat$.
\item[Conjunction]
Let
$ S_1\wedge S_2 \defeq \{(t;\s)\in\Lambda\times\S : (\pi_1(t);\s)\in S_1 \land (\pi_2(t);\s)\in {S_2} \big)\}$
for $S_1, S_2 \in\sat$. Then it is then straightforward to check that $\epair{e_1}{e_2}\defeq \lambda x.(e_1\,x,e_2\,x)$, where $\efst\defeq \pi_1$, 
$\esnd\defeq \pi_2$ define the expected evidences.
\item[Universal Implication] 
\newcommand{\myS}{\vec{S}}
For $S_1\in\sat$ and $\myS\in\P(\sat)$, 
we define the implication of propositions by
$S_1\efimp \myS \defeq 
\{(t;\s)\in\Lambda\times\S : \forall u. \big((u;\s)\in S_1 \limp (t\,u;\s) \in\bigcap_{S\in\myS}S\big)\}$.
Let $\lambda e \defeq \lambda xy.e\,(x,y)$
and $\eeval \defeq \lambda x. (\pi_1(x))\,(\pi_2(x))$.
Let $S_1,S_2\in\sat$ and $\myS\in\P(\sat)$ be saturated sets.
It is straightforward to show that  $\eeval \ureal (S_1 \efimp \myS) \wedge S_1 \to S$ for any $S\in\myS$.
We prove that if $e\in \Lambda$ is such that $(\forall S \in \myS.\; e \ureal S_1 \wedge S_2 \to  S) $
then $\lambda e \ureal S_1 \to (S_2 \efimp \myS)$.
Let $(t_1;\s)\in S_1$, then $\lambda e\,t_1 \reds \s \s \lambda y.e\,(x,y)$.
 Clearly, if $(t_2;\s)\in S_2$, 
 $\lambda y.e\,(x,y)\,t_2 \reds \s \s e\,(t_1,t_2)$.
 Since for any $S\in\myS$ the last term belongs to $S$, we can conclude by anti-reduction
 that $\lambda e \ureal S_1 \to (S_2 \efimp \myS) $.
 \qedhere
 \end{description}
 \end{proof}
\autoref{p:evidenced} implies, in particular, that our interpretation also induces a tripos
and a topos, by following the method described in~\cite{CohMiqTat21}.
In the following sections, we pay attention to nonstandard reasoning principles
for which we can define universal realizers, as these are the evidences
for our interpretation (as shown by \autoref{p:evidenced}).

\section{Nonstandard principles in realizability with slices}
\label{s:nonstandard} 
\subsection{Natural numbers}
\label{s:natural}
In \Cref{s:slices}, we considered a setting with a value restricted
variant of the $\nat{\cdot}$ predicate. 
Nonetheless, we can still assert that an expression 
is a natural number through the formula 

\[\natp e ~~\defeq~~ \forall X. (\bvnat e X) \to X.\]

As seen below, realizers of this formula will give access to the
expected computations for natural numbers.

\newcommand{\anyt}{\automath{\dagger}\xspace}
\begin{rem}
Observe that the language of HA2 does not express the existence
of specific nonstandard elements, e.g.\ $\delta$ is not in the language.
However, to refer to some nonstandard element $\tau$,
we can always consider a valuation that maps a variable $x$ to $\tau$. 
With abuse of notation,
in the remainder of this paper, we will write nonstandard elements directly 
in formulas as if they were in the language.
Also, we will use the notation \anyt to refer to an arbitrary $\lambda$-term 
with no further assumption.
\end{rem}

Using an argument similar to \autoref{r:exbv},
one can show that for any individual $f\in\N^\S$, 
if $t$ is a term such that $(t;\s)\in\tvS{\natp f}$,
then one can actually compute out of $t$ (without changing the value
of the state) the value of $f(\s)\in\N$.
In other words, $t$ is an effect-free term producing $\overline{f(\s)}$.
This is to be compared with $\nat f$, for which the requirement for its truth value to be saturated would have entailed its interpretation 
to reduce to a natural number $f(\s')$ in a possibly different state.

\begin{prop}\label{r:natp}  
  Let $f\in\N^\S$ and $\s\in\S$. If $t$ is a term such that $(t;\s)\in\tvS{\natp f}$,
then $t\,\lambda x.x \reds \s \s \overline{n}$, where $n=f(\s)$.
\end{prop}
\begin{proof}
 Let us define $\X\defeq \{(t;\s'): t\reds \s' \s \overline{n}\}$.
 This set is clearly saturated, and it is easy to see that 
 $(\lambda x.x;\s)\in\tvS{\bvnat f \X}$ (since $\lambda x.x\,\overline{n}\reds \s \s \overline{n}$).
 Therefore, we have that $t\in\tvS{(\bvnat f \X)\imp \X}$ and then
 $(t\,\lambda x.x;\s)\in\X$, that is $t\,\lambda x.x \reds \s \s \overline{n}$.
\end{proof}

We now show that (by-value) natural numbers, i.e.\ $\mathrm{Nat'}$, contain $0$, 
and are closed under the successor and recursion for internal formulas.

\begin{prop}\label{r:natural} 
Let $A$ be an internal formula.
We have

\begin{enumerate}
  \item $\lambda x.x\,0    \ureal \natp{0} $
  \item $\lambda xy.y\,(\suc\,x) \ureal \fabv{x}.\natp{S(x)}$
  \item $\rec \ureal  A(0)\imp \big(\fabv{y}.( A(y) \imp A(S\,y))\big) \imp \fabv{x}.A(x))$
\end{enumerate}
\end{prop}
\begin{proof}
Easy realizability proofs by anti-reduction.
{
 \begin{enumerate}
  \item Follows from the definition of $\natp{0}$: if $\X\in\sat$ is a saturated set,
  $\s$ a state and $t$ a term such that $(t;\s) \in \tvS{\bvnat {0} \X}$,
  we have $(\lambda x.x\,0)\,t~\reds \s \s~ t\,0\in\X$. Since $\X$ is saturated, we conclude by anti-reduction.
  \item 
  Let $f\in\N^\S$, $\X$ be a saturated set, $\s$ be a state 
  and $t$ be a term such that $(t;\s)\in\tvr{\bvnat {S f} \X}$.
  Let us write $n\defeq f(\s)$.
  Then 
  $(\lambda xy.y\,(\suc\,x))\,\overline{n}\,t\reds \s \s
  t\,(\suc\,\overline{n})$. Since $\suc\,\overline{n}=\overline{n+1}=\overline{S(f)(\s)}$,
  we get that $t\,(\suc\,\overline{n})\in\X$ and we conclude by anti-reduction.

  \item Directly from  Propositions \ref{r:rec_bv} and \ref{r:elimination}.\qedhere
 \end{enumerate}

 }
\end{proof}

The interpretation now witnesses the existence of new elements.
The canonical example is the \emph{diagonal}, i.e.\ the function $\delta:n\mapsto n$.
Indeed, the diagonal is a nonstandard natural number which is realized by the $\get$ instruction. 
We first show a lemma concerning the storage operator $T$ (from \autoref{d:T}) in this new context.

\begin{lem}\label{r:lem_T}
 Let $\s\in\S$ and $t,u$ be terms. 
 
 \begin{enumerate}
 \item\label{r:T}  For any  $n\in\N$,
 if $u \reds \s \s \overline{n}$,
  then $T\,t\,u \reds \s \s t\,\overline{n}$. 
\item \label{r:T_cor} For any $f\in\N^\S$, if $u \reds \s \s \overline{f(\s)}$
and $(t;\s)\in\tvS{\fabv{x}.A(x)}$,
 then $T\,t\,u \in\tvS{A(f)}$.
 \end{enumerate}
 \end{lem}
\begin{proof}
The first part is an easy induction on $n$, and the second part follows from the first by anti-reduction.
{
\begin{enumerate}
 \item  By induction on $n$.
  \begin{itemize}
\item If $n=0$, we have
 \begin{align*}
 T\,t\,u 
 &~\reds \s \s~
 \rec\,(\lambda y.y\,0)\,(\lambda xyz.y\,(\lambda x.z\,(\suc\,x)))\,u\,t \\
 &~\reds \s \s~ \rec\,(\lambda y.y\,0)\,(\lambda xyz.y\,(\lambda x.z\,(\suc\,x)))\,0\,t\\
 &~\reds \s \s~ (\lambda y.y\,0)\,t
 ~\reds \s \s~ t\;0.
 \end{align*}
\item If $\s=S(n)$, we have
 \begin{align*}
 T\,t\,u 
 &\reds \s \s~
 \rec\,(\lambda y.y\,0)\,(\lambda xyz.y\,(\lambda x.z\,(\suc\,x)))\,u\,t \\
 &\reds \s \s ~\rec\,(\lambda y.y\,0)\,(\lambda xyz.y\,(\lambda x.z\,(\suc\,x)))\,\suc\overline{n}\,t \\
 &\reds \s \s ~
 (\lambda xyz.y\,(\lambda x.z\,(\suc\,x)))\, \overline{n}\,(\rec\,(\lambda y.y\,0)\,(\lambda xyz.y\,(\lambda x.z\,(\suc\,x)))\,\overline{n})\,t \\
 &\reds \s \s ~
 (\rec\,(\lambda y.y\,0)\,(\lambda xyz.y\,(\lambda x.z\,(\suc\,x)))\,\overline{n})\,(\lambda x.t\,(\suc\,x)) \\
 &\reds \s \s~
 (\lambda x.t\,(\suc\,x))\,\overline{n} \reds \s \s t\,\suc\overline{n},
  \end{align*}
 where we used the induction hypothesis to obtain the penultimate reduction.
\end{itemize}
\item By definition, it holds that $(t;\s)\in\tvS{\bvnat f {A(f)}}$. 
 By part 1, we obtain that $T\,t\,u\reds \s \s t\,\overline{f(s)}$, 
 hence the result follows by anti-reduction.\qedhere
 \end{enumerate}}
 \end{proof}

\begin{prop}[ENS$_0$]\label{r:delta}

We have that

\begin{enumerate}
 \item $\anyt       \ureal \neg \st{\delta} $  
 \item $\anyt       \ureal \exists x.\neg \st{x}$
 \item $\lambda x.T\,x\,\get \ureal \natp{\delta}$
 \item $\lambda x.T\,x\,\get\,\dagger\ureal \exbv x.\neg \st{x}$
 \end{enumerate}
 \end{prop}
\begin{proof} \hfill 
\input{delta}
\end{proof}

Part 2 in \autoref{r:delta} is sometimes referred to 
as the ENS$_0$ (existence of nonstandard elements) principle~(e.g.\ in \cite{BerBriSaf12}). 
As a consequence of \autoref{r:elimination},
Leibniz equality is only compatible with the \fedrule-rule restricted to internal formulas.
In our setting, this encoding only reflects equality in the current state, i.e.\ a local
knowledge of individuals (slice by slice),
while the usual notion of equality (for $\N^\S$) requires a global knowledge (on all the slices).
If $A(x)$ is an external formula, we cannot hope to have an internal definition of equality
such that its elimination principle  $x=y \imp A(x) \imp A(y)$ is valid.

\begin{exa}\label{ex:equality} 
Consider  an individual $f$, equal to 1 everywhere except
for some state $\s_0$ where it is equal to $0$.

For any state $\s\neq \s_0$, we have $(\lambda x.x;\s)\Vdash 1^\inj = f$.
However, if we consider the formula $A(x)\defeq(\stto{x} \bot)\to \bot$,
then, for $\s\neq \s_0$, we have $(\lambda x.x\,\anyt;\s)\in\tv{A(1)}$ 
and $\tvs{A(f)}=\tvs{(\bot \to \top)\to\bot}$.
Thus, if $(t;\s)$ is a realizer of $$\forall Z.(Z(1^\inj)\imp Z(f)) \imp A(1^\inj)\imp A(f),$$
we immediately get that $(t (\lambda x.x) (\lambda x.x) (\lambda x.x\,\anyt);\s)\Vdash \bot$.
\end{exa} 
   
Nonetheless, the elimination of Leibniz equality is realizable
for standard individuals or for internal formulas.
 \begin{prop}\label{r:Leibniz} 
 Let $f$ and $g$ be individuals in $\N^\S$ and let $A(x)$ be a formula. Then
 
   \begin{enumerate} 
   \item $\lambda xyz.z \ureal \stto{f} \stto{g} (\forall Z.(Z(f)\to Z(g))) \to A(f) \to A(g)$
   \item If $A(x)$ is internal, then
   $\lambda x.x \ureal (\forall Z.(Z(f)\to Z(g))) \to A(f) \to A(g)$
   \end{enumerate}
 \end{prop}
 
\begin{proof}
 \begin{enumerate}
  \item If either $f$ or $g$ is not standard, the result is trivial.
  Assume that $f$ and $g$ are standard. The case $f=g$ is trivial, 
  and if $f\neq g$, we have  $\tvS{(\forall Z.(Z(f)\mapsto Z(g))}=\tvS{\top\mapsto\bot}$. 
  \item The result easily follows from \autoref{r:elimination}.\qedhere
  \end{enumerate}
\end{proof}

\subsection{Nonstandard reasoning principles}
In this section, we prove some properties which are usual in frameworks that use
nonstandard analysis: Transfer, Overspill, External Induction, Idealization, etc.

\autoref{r:transfer} below indicates that the Transfer property (for internal formulas)
is devoid of computational content.
This is a somewhat reassuring fact: 
properties that are true  for standard individuals are automatically true 
for all individuals.

\begin{thm}[Transfer]\label{r:transfer}
For any internal formula $A$ we have:
\vspace{-1em}
\begin{multicols}{2}
\begin{enumerate}
\item $\bigcap_{f \in\N^\S} \tvS{A}_{x\mapsto f} = \bigcap_{n\in\N} \tvS{A}_{x\mapsto n^\inj}$
\item $ \lambda xy.x  \ureal \forall x.A(x)\to\fas{x}.A(x)$
\item $ \lambda x.x\,\anyt \ureal \fas{x}.A(x)\to\forall x.A(x)$
\item $\bigcup_{f \in\N^\S} \tvS{A}_{x\mapsto f} = \bigcup_{n\in\N} \tvS{A}_{x\mapsto n^\inj}$
\item $ \lambda x. (\anyt,x) \ureal \exists x.A(x)\to\exs{x}.A(x)$
\item $ \lambda x.\pi_2(x) \ureal \exs{x}.A(x)\to\exists x.A(x)$
\end{enumerate}
\end{multicols}
\end{thm}
\begin{proof} 
Parts 1 and 4 follow from the glueing theorem. Indeed, we have:
$$
\begin{array}{rcl}
  \bigcap_{f \in\N^\S} \tvS{A}_{x\mapsto f} 
&=& \bigcap_{f \in\N^\S} \bigcup_{\s\in\S}\tvS{\trunc A}_{x\mapsto f} \times \{\s\} \\
&=& \bigcap_{f \in\N^\S} \bigcup_{\s\in\S}\tvS{\trunc A}_{x\mapsto (f(\s))^\inj} \times \{\s\} \\
&=& \bigcap_{n \in\N}    \bigcup_{\s\in\S}\tvS{\trunc A}_{x\mapsto n^\inj} \times \{\s\} \\[0.3em]
&=& \bigcap_{n\in\N} \tvS{A}_{x\mapsto n^\inj}
\end{array}
\eqno
\text{
\begin{tabular}{r}
 (by glueing)\\
 (by def. of $\trunc{\,\cdot\,}$)\\
 \\[0.3em]
 (by glueing)\\
\end{tabular}
}
$$
The proof for part 4 is analogous.

Parts 2 and 3 (resp. 5, 6) are direct consequences of the first (resp. fourth) part.
For instance, for part 3, let $\s$ be a state and $u$ be a term such that $(u;\s)\in\tvS{\fas{x}.A(x)}$.
Recalling that $\tvS{\st{n^\inj}}=\Lambda\times\S$ for any $n\in\N$, we have:
\[
\begin{array}{rcl}
\forall f\in\N^\S, v\in\Lambda. (v;\s)\in \tvS{\st f} &\Rightarrow& (u\,v;\s)\in\tvS{A(x)}_{x\mapsto f}\\
&\Rightarrow& \forall n\in\N, v\in\Lambda. (u\,v;\s)\in\tvS{A(x)}_{x\mapsto n^\inj} \\
&\Rightarrow& \forall v\in\Lambda. (u\,v;\s)\in\bigcap_{n\in\N}\tvS{A(x)}_{x\mapsto n^\inj} \\
&\Rightarrow& \forall v\in\Lambda. (u\,v;\s)\in\bigcap_{f\in\N^S}\tvS{A(x)}_{x\mapsto f}
\end{array}
\]
where the last implication is obtained using part 1.
In particular, $(u\,t;\s)$ belongs to $\tvS{\forall x.A(x)}$ and by anti-reduction,
so does $((\lambda x.x\,t)u;\s)$. \qedhere

\end{proof}

As expected, Transfer does not hold for all formulas. A counter-example is given in the next proposition by the external formula stating that all individuals are (not not) standard.
 \begin{prop}\label{prop:nonst}
  Let $A(x)$ denote the formula $\neg\st{x}$. 
  Then, there is no realizer for the formulas  $\fas{x}.\neg A(x)\rightarrow\forall x.\neg A(x)$
  and $\exists x.A(x) \to \exs{x}. A(x)$.
 \end{prop}
 \begin{proof} 
Both statements follow from the definitions. For instance, for the second formula, 
observe that
$$\bigcup_{f\in\N^\S}\{(t;\s): (\pi_1(t);\s) \in \tvS{\st f} \land  (\pi_2(t);\s)\in \tvS{\neg \st f}\}=\emptyset$$
since for any $f\in\N^\S$, either $\tvS{\st f}$ or $\tvS{\neg \st f}$ is empty.
Consequently, we have $\tvS{\exs{x}.A(x)}=\emptyset$ while
$\tvS{\exists x.A(x)}=\tvS{\top}=\Lambda\times\S$.
\end{proof}

The principle of External Induction \cite{NelsonBookREPT} allows to prove that a certain property 
is valid for all standard natural numbers. For instance, the assertion stating that every nonstandard
element is larger than all standard natural numbers\footnote{Actually,
this requires to consider a quotiented definition of the standardness predicate, see \autoref{r:st_lower}.}.
We show that in our context, this principle can be realized using the $\rec$ instruction.

\begin{prop}[External induction]\label{r:external_induction} 
For any formula $A(x)$ whose only free variable is $x$\nomidem
\[ \rec\ureal A(0^\inj) \imp \fabvs{x}.(A(x) \imp A(S(x))) \imp \fabvs{x}.A(x). \] 
\end{prop}
\begin{proof}

Let $\s$ be a state, $n\in\N$ be a natural number
and $u_0$, $u_S$ be terms and $V$ be a value such that
$(u_0;\s)\in \tvS{A(0^\inj)}$, $(u_S;\s)\in \tvS{\fars y.(A(y)\imp A(S(y))}$
and $(V;\s)\in\tvS{\nat{n^\inj}}$.
The latter implies that $V=\overline{n}$.
Let us prove, by induction on $n$, that 
\[ \rec\,u_0\,u_S\,\overline{n}\in \tvS{A(n^\inj)} \]
\begin{itemize}
\item If $n=0$, then we have that $\rec\,u_0\,u_S\,\overline 0 \reds{s}{s} u_0$, the result
follows by anti-reduction from the hypothesis on $u_0$.
\item If $n=S(m)$, then we have that $\rec\,u_0\,u_S\,(\suc\, \overline m)  
\reds{s}{s} u_S\,\overline m\,(\rec\,u_0\,u_S\,\overline m)$.
By induction hypothesis, we have that $(\rec\,u_0\,u_S\,\overline m;\s)\in \tvS{A(m)}$.
The result thus follows (by anti-reduction) from the hypothesis on $u_S$.
\qedhere
\end{itemize}
\end{proof}

 The next two propositions show
 that one cannot separate standard natural numbers from nonstandard natural numbers
 using an internal formula~\cite{Robinson66}. This fact is usually formalized by the properties of Overspill and Underspill. 
 We first show that, in our context, Overspill can be \emph{realized} by combining
 the realizers for ENS$_0$ and for the Transfer principle.
\begin{prop}[Overspill]\label{r:Overspill} 
For any internal formula $A$, we have\nomidem

\[ \lambda x.(\anyt,x\,\anyt)\ureal \fas x.A(x) \imp \exists x.(\neg \st{x} \land A(x)). \]
\end{prop}

{
\begin{proof}
 Let $(u;\s)\Vdash \fas x.A(x)$.
 Let us show that $((\lambda x.(t,x\,t))u;\s)\Vdash \exists x.(\neg \st{x} \land A(x))$.
 Following the proof of part 3 in \autoref{r:transfer}, we obtain $(u\,t;\s) \Vdash \forall x.A(x)$
 and therefore $(u\,t;\s) \Vdash A(\delta)$.
 By ENS$_0$ (\autoref{r:delta}), we have $(t;\s) \Vdash \neg \st \delta$.
 Finally, we obtain that $((t,u\,t);\s)\Vdash \exists x.(\neg \st x \land A(x))$
 and we can conclude by anti-reduction.
\end{proof}
}

The usual proof of Underspill is by contradiction, hence using classical logic, which we do not have here. Nevertheless, we can obtain the following version in which a double-negation occurs.
\begin{prop}[Underspill]\label{r:Underspill} 
For any internal formula $A$, we have

\[ \lambda xy. (\lambda z.y\,(\anyt,z))(x\,\anyt)\ureal (\forall x.\neg \st x \imp A(x)) \imp \neg\neg \exs x.A(x) .\]
\end{prop}
{
\begin{proof}
Let $\s$ be a state, and $u$, $v$ be terms such that $(u;\s)\Vdash \forall x.\neg \st x \imp A(x)$
and $(v;\s)\Vdash \neg \exs x.A(x)$.
Using the adequacy of congruence rules (\autoref{r:cong_states}), 
observe that $(v;\s)\Vdash \forall x.((\st x \land A(x))\to \bot)$,
and by currying
\[ (\lambda wz.v\,(w,z);\s)\Vdash \fas x. A(x)\to \bot \]
Since $A$ is internal, by Transfer, we get
\[ (\lambda z.v\,(t,z);\s)\Vdash \forall x. A(x)\to \bot \]
By the hypothesis on $u$ and ENS$_0$, we have $(u\,t;\s)\Vdash A(\delta)$,
hence 
\[ (\lambda z.v\,(t,z))(u\,t);\s)\Vdash \bot, \]
and we can conclude by anti-reduction.
\end{proof}
}

\subsection{Idealization}
\label{s:idealization}
We first extend the realizability interpretation to take into 
account relations $R:\N^2\to\N$ on the natural numbers:
\[
\begin{array}{rcl}

 \tvr{R(e_1,e_2)}   & \defeq &
 \{(t;\s): R(\foint{e_1}_{\rho}(\s),\foint{e_2}_{\rho}(\s)) \text{ holds}\}.
 \end{array}
 \]
 This coincides with the interpretation of the relation $R$ through a second-order variable and the corresponding semantic relation from $\N^2$ to $\sat$ in the interpretation.
 
 Let us now briefly illustrate the main idea behind the proof of Idealization
by showing that there exists a (nonstandard) natural number greater than or equal to
any standard number. 
The usual proof relies on the fact that $\delta$ is such a number, since for any standard 
number $n$,
in any slice 
greater than or equal to $n$, the relation $n\leq \delta$ holds.
In our setting, we use the $\set$ instruction to reach such a state.

\begin{prop}[Diagonalization]\label{r:diag}
We have $\lambda z.T\,z\,\get\,(\lambda xy.\set\,y\,\anyt)\ureal \exbv x .\fabvs{y}. y\leq x$.
\end{prop}
\begin{proof}
\input{diag}
\end{proof}

Consider a term $\loopp$ such that\footnote{For instance, 
we can define $\loopp \defeq Y\,\incr$, where $Y$ is 
the usual fixed-point combinator of the $\lambda$-calculus.}
for any state $\s\in\S$ it holds that 
$\loopp \reds \s \s  \incr\,\loopp$, where  $\incr\defeq \lambda x.\set\,(\suc\,\get)\,x$. 
 Then for any natural number $n\in\N$ and any state $\s\in\S$, 
$\loopp \reds \s {\s'}  \loopp$ where $\s'\ge n$.
Since for any $\s'\ge n$, $(\anyt;\s')\in\tvS{n<\delta}$, by anti-reduction we obtain
the following Proposition.
\begin{prop}\label{rmk:diag}
We have $\lambda w.\loopp\ureal \fas x. x<\delta$.
\end{prop}
Observe that here the value of $n$ is not required,
so the quantifier does not need to be relativized.
Yet, the computation never terminates and 
we do not even know when the computation reaches a correct state.

As mentioned above, the idea to prove the general case of Idealization is very similar.
If for any $n\in\N$ there exists $\tau_n\in\N$ such that for any $m\leq n$, $R(\tau_n,m)$ holds,
we can consider the nonstandard natural number $\tau\defeq (\tau_\s)_{\s\in\S}\in\N^\S$.
Using a witness extraction mechanism, as provided by the next proposition,
we can compute $\tau$ from any realizer of 
$\fabvs{n}.\exbvs{x}.\fabvs{y}.(y\le n\imp  R(x,y))$.

\begin{prop}[Witness extraction]\label{r:extraction}
 For any formula $A$, any valuation $\rho$ closing $\exists x.A$, any state $\s$ and any term $t$
 such that $(t;\s)\in\tvr{\exbv{x}.A}$, there exists a natural number $f\in\N^\S$
 and a term $u$ such that $(u;\s)\in\tvS{A}_{\rho,x\mapsto f}$ 
 and  $t\,(\lambda xy.(x,y))\reds{\s}{\s} (\overline{f(\s)},u)$.
\end{prop}
\begin{proof}
Assume that $(t;\s)\in\tvr{\exbv{x}.A}$.
 By definition,  for any $\X\in\sat$ and any $(v;\s)\in \tv{\fabv x.(A  \imp X)}_{\rho,X\mapsto\X}$, we have that
 $(t\,v;\s)\in\X$.
 Let us define the set 
 \[ \X \defeq \{(w;\s')\in\Lambda\times\S:\exists f\in\N^\S. \exists u\in\Lambda.
 ~w\reds{\s'}{\s} (\overline{f(\s)},u)~ \land~ (u;\s)\in\tvS{A}_{\rho,x\mapsto f}\}, \]
 which is obviously saturated.
 Clearly $(\lambda xy.(x,y);\s)\in\tvS{\fabv x.(A  \imp X)}_{\rho,X\mapsto\X}$
 since for any $f\in\N^\S$ and any $(u;\s)\in\tvS{A}_{\rho,x\mapsto f}$, it holds that
 $(\lambda xy.(x,y))\,\overline{f(\s)}\,u \reds \s \s (\overline {f(\s)},u)\in\X$.
 We conclude that 
 $(t\,(\lambda xy.(x,y));\s)\in\X$, i.e.\
  $t\,(\lambda xy.(x,y)) \reds \s \s (\overline {f(\s)},u)$.
\end{proof}

The term 
\[\mathsf{ideal}\defeq \lambda x.\lambda y.T\,y\,(\pi_1(T\,(x\,\anyt)\,\get\,(\lambda xy.(x,y))))\,(\lambda yz.\set\,z\,y)\]
 is a realizer for the Idealization principle. 
Indeed, in any state $\s$ the first component of $\mathsf{ideal}$ 
computes $\tau(\s)$ (using \autoref{r:extraction}),
while the second component increases the state to ensure the validity of the relation (as in \autoref{r:diag}). 

\begin{thm}[Idealization]\label{r:idealization}
We have:
\[\mathsf{ideal}\ureal\fabvs{n}.\exbv{x}.\fabvs{y}.(y\le n\imp  R(x,y)) \imp \exbv x .\fabvs{y}. R(x,y).\]
\end{thm}

\begin{proof}
Let $\s$ be any state and $u$ a term such that 
$$(u;\s)\in \tvS{\fabvs{n}.\exbv{x}.\fabvs{y}.(y\le n\imp  R(x,y))}.$$
By part~\ref{r:T_cor} of Lemma~\ref{r:lem_T}, this entails that 
\[ (T\,(u\,\anyt)\,\get;\s)\in \tvr{\exbv{x}.\fabvs{y}.(y\le \s\imp  R(x,y))}. \]

 By \autoref{r:extraction}, we know that 
 there exists a natural number $f_\s\in\N^\S$ and a term $v_\s\in\Lambda$
 such that 
 $T\,(u\,\anyt)\,\get\,(\lambda xy.(x,y))\reds{\s}{\s} (\overline{f_\s(\s)},v_\s)$ 
 and $(v_\s;\s)\in\tvS{\fabvs{y}.(y\leq \s\imp  R(f_\s,y))}$.
The latter implies that for any $m\in\N$ such that $m\leq \s$ 
 and any term $t$, it holds that $(v_\s\,t\,\overline{m}\,t;\s)\in \tvS{R(f_\s,m)}$
 and hence $R(f_\s(\s),m)$ holds (since $\tvs{R(f_\s,m)}=\{(t;\s): R(f_\s(\s),m) \text{ holds}\}$).
 
Consider the (possibly nonstandard) individual $\tau \in\N^\S$
 defined by $\tau(\s)=f_\s(\s)$ .

  We have
  \[ \mathsf{ideal}\;u ~\reds \s \s ~\lambda y.T\,y\,(\pi_1(T\,(u\,\anyt)\,\get\,(\lambda xy.(x,y))))\,(\lambda yz.\set\,z\,y) \]
  hence, by part~\ref{r:T_cor} of Lemma~\ref{r:lem_T}, to conclude by anti-reduction it suffices to prove that

  \begin{enumerate}
  \item \underline{$\pi_1(T\,(u\,\anyt)\,\get\,(\lambda xy.(x,y)))\reds \s \s \overline{\tau(\s)}$}.
  Indeed, we know that this term reduces as follows:
  \[\pi_1(T\,(u\,\anyt)\,\get\,(\lambda xy.(x,y))) \reds{\s}{\s}
    \pi_1(\overline{f_\s(\s)},v_\s) \reds{\s}{\s} \overline{f_\s(\s)}\]
  and by definition $\tau(\s)=f_\s(\s)$.
  \item  \underline{$(\lambda yz.\set\,z\,y;\s)\Vdash \fabvs{y}. R(\tau,y)$}.
  To prove this, it suffices to show that for any $m\in\N$ and any $t\in\Lambda$, 
  we have $((\lambda yz.\set\,z\,y)\,t\,\overline{m} ;\s)\Vdash R(\tau,m^\inj)$.
  With $\s'\defeq \max(\s,m)$, we have that $(\lambda yz.\set\,z\,y)\,t\,\overline{m} \reds \s \s 
  \set\,\overline{m}\,t \reds \s {\s'} t$.
  By construction, since $m\leq \s'$, we know that $R(\tau(\s'),m)$ holds,
  hence $(t;\s')\in\tvr{R(\tau(\s'),m)}$ and we conclude by anti-reduction.\qedhere
  \end{enumerate}
\end{proof}

\section{LLPO}
\label{s:LLPO}

In this section we give a realizer for a nonstandard version of the 
\emph{Lesser Limited Principle of Omniscience}:
\[
 \textrm{LLPO}:=\forall x.\forall y. (A(x) \lor B(y)) \to (\forall x.A(x) \lor \forall y.B(y))
\]
This principle is a semi-intuitionistic principle, in the sense that it is seen 
as being nonconstructive (it is indeed provably false in some intuitionistic theories, 
{\cf} \cite[p.~4]{BridgesRichman})
while still being weaker than the full law of excluded middle. 

\subsection{LLPO in nonstandard arithmetic}
\label{s:llpo_intro}
\newcommand{\tvrq}[1]{\tv{#1}_\rho^*}

We will consider a variant of the LLPO principle in our setting, where the quantifiers are
restricted to standard elements and the formulas $A$ and $B$ are internal 
(where $x$ (resp. $y$) does not occur in $B$ (resp. in $A$)):

\[
 \textrm{\llpost}:=\fabvs x.\fabvs y. (A(x) \lor B(y)) \to (\fabvs x.A(x) \lor \fabvs y.B(y))
\]
Let us give an overview of our computational interpretation for this principle,
which will rely on the several realizers introduced in \Cref{s:llpo_real}
and described in \autoref{fig:LLPO}.
Assume that we are given, in a certain state, a realizer of the hypothesis
\[H_{A,B} \defeq \fabvs{x}.\fabvs{y}. (A(x) \vee B(y)).\]
The main idea consists in turning this term into a realizer of 
\[\fabvs{z}.(A_{\leq z} \lor B_{\leq z}),\]
where  $A_{\leq z}\defeq \fabv{x}.x\leq z \to A(x)$.

Indeed, observing that the formula $A_{\leq z}$ is internal, by Transfer and 
instantiation with $\delta$ (or any other nonstandard element), 
the proposition $x\leq\delta$ becomes trivially true for any standard $x$ and 
we get the expected conclusion
\[(\fabvs{x}.A(x)) \lor (\fabvs{y}. B(y)).\]

In fact, this last step is the only step where we actually use nonstandard principles 
(here Transfer and the existence of nonstandard elements).
The rest of the proof, forgetting all the relativizations to standard elements,
would be valid in standard arithmetic. This is reflected by the fact
that we only use External Induction and properties of the disjunction.
In terms of realizers, this means that we will only use universal realizers
that will never manipulate the state.

To get a realizer of $\fabvs{z}.(A_{\leq z} \lor B_{\leq z})$, we rely on External Induction (as the term $t_{\mathrm{aux}}$ shows),
the main difficulty lying in proving the induction step
\[\fabvs{x}.\left(A_{\leq x}\lor B_{\leq x} \to A_{\leq S(x)}\lor B_{\leq S(x)}\right).\]

To illustrate this step, let us consider the case where $A_{\leq x}$ holds.
To obtain the expected conclusion, it is sufficient to show that $A(S(x))\lor B_{\leq S(x)}$ holds.
This leads us to break the symmetry between $A$ and $B$ by considering the formula
\[
\Phi_{A,B}(x,y) \defeq A(x)\lor B_{\leq y}.
\]

But using our starting assumption, namely a realizer of $\fabvs{x}.\fabvs{y}. (A(x) \vee B(y))$,
for any standard $x$ we can easily get $\fabvs{y}.\Phi_{A,B}(x,y)$ by external 
induction, and thus $\fabvs{x}.\fabvs{y}.\Phi_{A,B}(x,y)$ which is enough to conclude
the whole proof. 
For the inductive step $\fabvs{y}.\left(\Phi_{A,B}(x,y) \to \Phi_{A,B}(x,S(y))\right)$ of 
the latter induction, we reason by cases on the induction hypothesis:
\begin{itemize}
 \item if $A(x)$ holds then the conclusion follows immediately,
 \item if $B_{\leq y}$ holds, then we use the assumption to get either $A(x)$ or $B(S(y))$,
 and again, in both cases the conclusion follows.
 \end{itemize}

This proof is a variation of \cite[Prop. 3.4]{BerBriSaf12},
the main difference being that in our context, 
we have access to concrete nonstandard elements (namely $\delta$), 
and we can instantiate a certain formula with $\delta$ 
instead of using the Idealization principle.

\subsection{Disjunction}
In order to define a realizer for \llpost, we first need to extend our language with disjunctions.
We choose to rely on a primitive disjunction rather than on a second-order impredicative encoding 
of disjunction as the latter would make the task of finding realizers much more difficult without
bringing additional strength to our setting.

We thus extend the languages of terms and formulas as follows:
\[\begin{array}{lr@{~~}c@{~~}l}
\text{\bf Formulas} & A,B   &::= & ...  \mid A \lor B\\
\text{\bf Terms   } & t,u   &::= & ... \mid  \inl t \mid \inr t \mid \mw{t}{t_1}{t_2}
\end{array}
\]
and the type system accordingly
\begin{mathpar}
 \infer[\oriurule]{\TYP{\Gamma}{\inl t   }{A_1 \lor A_2}}{\TYP{\Gamma}{t}{A_1}} 

    \infer[\oridrule]{\TYP{\Gamma}{\inr t   }{A_1 \lor A_2}}{\TYP{\Gamma}{t}{A_2}} 

    \infer[\orerule]{\TYP{\Gamma}{\matchwith{t}{x_1}{t_1}{x_2}{t_2}}{C}}{\TYP{\Gamma}{t}{A_1 \lor A_2}
    & \TYP{\Gamma,x_i:A_i}{t_i}{C}}
\end{mathpar}
We also extend the reduction system with one extra case to define contexts
\[
 C[] ::= ... \mid \mw{\hole}{t_1}{t_2}
 \]
and one additional reduction rule for this new operations
\[
 \axinf{\mw{\iota_i(t)}{t_1}{t_2}}{t_i[t/x_i]}
\]

Finally, we extend the realizability interpretation to include the case of disjunction. 
We base the definition on the elimination rule of the disjunction (see \autoref{rmk:elim}),
as has been done before for the other connectives:
\[
\begin{array}{rcl}
\tvr{A_1\lor A_2     }& \defeq & 
\{(t;\s)\in\Lambda \times \S : \forall t_1,t_2,S\in\sat.\\
&\multicolumn{2}{l}{\left(\forall i\in\{1,2\}. 
\forall (u_i;\s)\in \tvr{A_i}.(t_i[u_i/x_i];\s)\in S\right) \limp} \\
&& \hspace{4cm}(\mw{t}{t_1}{t_2}; \s)\in S\}  \\

 \end{array}
\]
Observe that glueing still holds by simply defining
$\trunc{A \lor B } \,\defeq \, {\trunc{A}} \lor {\trunc{B}}$.

Once more, we take advantage of the modularity of the realizability interpretation
to get the adequacy with respect with the type system extended with disjunction
by only proving the adequacy of the new typing rules (see \autoref{def:adequate_state}).
\begin{prop}
 The rules $\oriurule$, $\oridrule$ and $\orerule$ are adequate.
\end{prop}
\begin{proof}
The adequacy of the rule $\orerule$ follows directly from the definition, by
considering the particular set $S=\tvr{C}$.

We now prove adequacy of the rule $\oriurule$. Assume that the typing judgment $\Gamma \vdash t:A_1$ is adequate 
with respect to some state $\s\in\S$.
To prove that the conclusion is adequate with respect to the same state, 
let us consider $\rho$ a valuation closing $A_1\lor B_2$ and $\Gamma$,
$\sigma$ a substitution such that $(\sigma;\s)\Vdash \rho(\Gamma)$ and 
let $t_1, t_2$ be two terms and $S\in\sat$
be such that for any $(u_i;\s)\in \tvr{A_i}$ we have that $(t_i[u_i/x_i];\s)\in S$.
Since $\sigma(\inl t) = \inl{\sigma(t)}$, we have
\[\mw{(\sigma(\inl{t}))}{t_1}{t_2} ~\reds \s \s ~t_1[\sigma(t)/x_1].\]
Using the hypotheses, we have that $(\sigma(t);\s)\in\tvr{A_1}$
and therefore $(t_1[\sigma(t)/x_1];\s)\in S$. 
We can conclude by anti-reduction.
\end{proof}

As an illustration of the use of disjunction, we define below a term 
allowing us to commute $A$ and $B$ in the premise of \llpost.
This term will be useful afterwards since the proof mostly relies on two External Inductions
in which the formulas $A$ and $B$ have asymmetric roles.

 \begin{lem}\label{r:t_lor}
  For any formulas $A$ and $B$ we have
\[ t_\lor \ureal \left(\fabvs{x}.\fabvs{y}.(A(x)\lor B(y) )\right) \imp \fabvs{x}.\fabvs{y}.(B(x)\lor A(y) )\]
 where
 $t_\lor \defeq \lambda hzxwy.\mww{h}{(h\,z\,x\,w\,y)}{\inr{h_1}}{\inl{h_2}}$.

 \end{lem}
 \begin{proof}

  Let $\s\in\S$ be a state,  $(h;\s)\Vdash \fabvs{x}\fabvs{y}.(A(x)\lor B(y) )$ and 
  $n,m \in\N$ be two natural numbers.
  We have 
  \[
   t_\lor\,h\,\anyt\,\bar n \, \anyt\, \bar m 
   ~\reds \s \s~ 
   \mww{h}{(h\,\anyt\,\bar n\,\anyt\,\bar m)}{\inr{h_1}}{\inl{h_2}}   
  \]
  The assumption on $h$ gives us that $(h\,\anyt\,\bar n\,\anyt\,\bar m;\s)\in\tvr{A\lor B}$.
  Since it is clear that for any $(t_A;\s)\in\tvr{A}$, 
  $(\inr{t_A};\s)\in\tvr{B\lor A}$ (and vice-versa with $(t_B;\s)\in\tvr{B}$
  and $\inl{t_B}$), by definition of $\tvr{A\lor B}$
  we have that the right-hand side terms belongs to $\tvr{B\lor A}$
  and we can conclude by anti-reduction.
  \end{proof}

\subsection{A realizer for \llpost}
\label{s:llpo_real}

\begin{figure}[t!]
\input{llpo.tex}
\caption{Terms used to realize \llpost}
\label{fig:LLPO}
\end{figure}
We shall now detail the definition of a realizer for \llpost. Its definition
follows the intuition of the proof sketched in \Cref{s:llpo_intro}.
To that end, we require several terms meant to be realizers corresponding
to the different steps of the proof, as described in \autoref{fig:LLPO}. 

Recall that the ordering $\cdot < \cdot$ 
and $\cdot \le \cdot$ are interpreted as any relations
$R:\N^2 \to \N$ following the definition given in \Cref{s:idealization}:
\[\begin{array}{rcl}

 \tvr{R(e_1,e_2)}   & \defeq &
 \{(t;\s): R(\foint{e_1}_{\rho}(\s),\foint{e_2}_{\rho}(\s)) \text{ holds}\}.\\
\end{array}
 \]

{We start by showing how a realizer of $A_{\leq n}$ can be obtained
from realizers of the different $A(m)$ for $m\leq n$.} 
\begin{lem}\label{r:t_leq0}
For any internal formula $A$, we have
 \[t_{\leq0} \ureal A(0) \to A_{\leq 0},\]
 where 
 ${ t_{\leq0} }\defeq \lambda xny.x$.
\end{lem}
\begin{proof}
Recall that $A_{\leq 0} \defeq \fabv{x}.x\leq 0 \to A(x)$.
Let $\s\in\S$ be a state, $f\in\N^\S$ be an individual 
and $u$, $v$ be terms such that
$(u;\s)\in\tvr{A(0)}$, $(v;\s)\in\tvr{f\leq 0}$.
In particular, the latter entails that $f(\s)=0$.
Therefore, since $A$ is internal, we get that $(u;\s)\in \tvr{A(f)}$ by glueing.
By construction, we have that
\[
 t_{\leq0} \, u \, \bar 0 \, v ~\reds \s \s~ u
\]
and we can conclude by anti-reduction.
\end{proof}

In the next lemma, we write $\ifte{\bar n}{\bar m}{t}{u}$ 
(where $n,m\in\N$ and $t,u\in\Lambda$) for a term that 
reduces to $t$ if $n=m$ and to $u$ otherwise (defining such a term
using the $\rec$ operator is an easy programming exercise, 
which we would rather not bother the reader with).

\begin{lem}\label{r:t_leqsuc}
For any internal formula $A$ and any natural number $n\in\N$, we have
 \[t_{\leq\suc}\,\bar n\ureal A_{\leq n} \imp A(S(n)) \imp A_{\leq S(n)},\]
 where ${ t_{\leq\suc} }\defeq 
 {\lambda nxymz.\ifte{m}{\suc n}{y}{(x\,m\,z)}}$.
\end{lem}
\begin{proof}
Recall that $A_{\leq n} \defeq \fabv{x}.x\leq n \to A(x)$.
Let $s\in\S$ be a state, $n\in\N$ be a natural number, $f\in\N^\S$ be an individual,
and $u,v,w\in\Lambda$ be terms such that
$(u;\s)\in\tvr{A_{\leq n}}$, $(v;\s)\in\tvr{A(S(n)))}$ and 
$(w;\s)\in\tvr{f\leq S(n)}$.
Putting $m \defeq f(\s)$, the latter entails that $m\leq S(n)$.
By construction, we have
 \[
  t_{\leq\suc}\,\bar n\,u\,v\,\bar m\,w ~\reds \s \s ~
  \ifte{\bar m}{\suc \bar n}{v}{(u\,\bar m\,w)}.
 \]
Let us reason by case analysis: 
\begin{itemize}
 \item if $f(\s)=m=S(n)$, then we have
 \[
  \ifte{\bar m}{\suc \bar n}{v}{(u\,\bar m\,w)}~\reds \s \s~v,
 \]
 and since $A$ is internal, we get that $(v;\s)\in \tvr{A(f)}$ by glueing
 which allows us to conclude by anti-reduction.
\item if $f(\s)=m<S(n)$, then we have
 \[
  \ifte{\bar m}{\suc \bar n}{v}{(u\,\bar m\,w)}~\reds \s \s~u\,\bar m\,w.
 \]
 By assumption on $u$, we have that $(u\,\bar m\,w;\s)\in\tvr{A(m)}$,
 and therefore $(u\,\bar m\,w;\s)\in\tvr{A(f)}$ using glueing.
 We can thus conclude by anti-reduction.\qedhere
 \end{itemize}
\end{proof}

If something is true below a certain nonstandard element, such as $\delta$, 
then it is true for any standard element.
This is connected with \autoref{rmk:diag} that states that $\delta$ is greater 
than any standard natural number, and is somewhat trivial in usual nonstandard settings 
(which is reflected here by the fact that the realizer is making a blind loop).

\begin{lem}\label{r:t_delta}

 For any formula $A$, we have 
 \[t_\delta\ureal A_{\leq\delta}  \to \fabvs{y}.A(y), \]
 where
 ${t_\delta}\defeq {\lambda xy.x\,y\,\inr \loopp}$.
\end{lem}
\begin{proof}
Let $\s$ be a state and $u$ be a term such that $(u;\s)\real \fabv{y}.y\leq \delta \imp A(y)$.
We need to show that $(t_\delta\,u;\s)\real \fabvs{y}.A(y)$.
Letting $n\in\N$ be a standard natural number,
we have 
\[t_\delta\,u\,\bar{n} \reds \s \s u\,\bar n\,\inr \loopp.\]
By \autoref{rmk:diag}, we know that $(\inr \loopp;\s)\real n<\delta$, 
hence $ u\,\bar n\,\inr \loopp \real A(n)$ (using the hypotheses on $u$) and we can conclude by anti-reduction.
\end{proof}

We now show how to build the different terms necessary to the first External Induction,
allowing us to obtain a realizer for the formula
$\fabvs{x}.\fabvs{y}.\Phi_{A,B}(x,y)$.

\begin{lem}\label{r:t_0}For any internal formulas $A$ and $B$, we have
\[t_0 \ureal \left(\fabvs{x}.\fabvs{y}.(A(x)\lor B(y) )\right) \to \fabvs{x}.\Phi_{A,B}(x,0),\]
where 
${t_0 }\defeq {\lambda {h}wx.\matchwith {({h}\,\anyt\,x\,\anyt\,0)} 
a {\inl{a}} b {t_{\leq0}\,b}}$.
\end{lem}
\begin{proof}

Recall that $\Phi_{A,B}(x,y)=A(x) \lor B_{\leq y}$.
Let $\s\in\S$ be a state, $n\in\N$ be a natural number and 
$h\in \Lambda$ be a term such that 
$(h;\s)\in\tvr{\fabvs{x}.\fabvs{y}.(A(x)\lor B(y)}$.
We have
\[
 t_0 \,h\,\anyt\,\bar n ~\reds \s \s ~
 \matchwith {({h}\,\anyt\,\bar n\,\anyt\,0)}a{\inl{a}}b{t_{\leq0}\,b}.
\]
Using the assumption on $h$, we have 
$(({h}\,\anyt\,\bar n\,\anyt\,0);\s)\in\tvr{A(x)\lor B(0)}$.
We can thus conclude by anti-reduction using
the adequacy of the $\orerule$ rule
and \autoref{r:t_leq0} for the $\inr \cdot$ case. 
\end{proof}

\begin{lem}\label{r:t_suc}
For any internal formulas $A$ and $B$, we have
\[t_\suc \ureal \left(\fabvs{x}.\fabvs{y}.(A(x)\lor B(y) )\right) \to
\fabvs{x}.\fabvs{y}.(\Phi_{A,B}(x,y) \to \Phi_{A,B}(x,\suc y)),\]
where 
$$\begin{array}{r@{}l}
{ t_\suc }\defeq \lambda {h}wxzyp.
\matchop\, p\, \{&\inl {p_1}\mapsto\inl{p_1}\\
&\inl {p_2}\mapsto\matchwith{({h}\,\anyt\,x\,\anyt\,(\suc y))}a {\inl{a}} b{
 \inr{t_{\leq\suc}\,p_2\,b}}. 
  \end{array}$$

\end{lem}
\begin{proof}

Recall that $\Phi_{A,B}(x,y)=A(x) \lor B_{\leq y}$.
Let $\s\in\S$ be a state, $n,m\in\N$ be natural numbers and 
$h,u\in \Lambda$ be terms such that 
$(h;\s)\in\tvr{\fabvs{x}.\fabvs{y}.(A(x)\lor B(y)}$
and 
$(u;\s)\in\tvr{\Phi_{A,B}(n,m)}=\tvr{A(n) \lor B_{\leq m}}$.
By construction, we have
\[
 { t_\suc } \, h \, \anyt\,\bar n\,\anyt\,\bar m\,u
 ~\reds \s \s~
 \begin{array}{r@{}l}
\matchop\, u\, \{&\inl {p_1}\mapsto\inl{p_1}\\
&\inl {p_2}\mapsto\matchwith{({h}\,\anyt\,\bar n\,\anyt\,(\suc \bar m))}
    a {\inl{a}} b {\inr{t_{\leq\suc}\,p_2\,b}}. 
  \end{array}
\]
To conclude by anti-reduction, we need to prove that the term 
on the right-hand side is in $\tvr{A(n)\lor B_{\leq S(m)}}$,
using the assumption on $u$ and the adequacy of the $\orerule$ rule.
The $\inl{\cdot}$ case is immediate. For the $\inr{\cdot}$ case,
let us consider a term $u_2$ such that $(u_2;\s)\in\tvr{B_{\leq m}}$ 
and prove that $\matchwith {({h}\,\anyt\,\bar n\,\anyt\,(\suc \bar m))}
    a {\inl{a}} b {\inr{t_{\leq\suc}\,u_2\,b}}\in\tvr{\Phi_{A,B}(n,S(m))}$.

Again, we use the assumption on $h$ and the adequacy of the $\orerule$ rule 
to conclude. Let us consider a term $b$ such that $(b;\s)\in\tvr{B(S(m))}$. 
It then follows from the assumption on $u_2$ and \autoref{r:t_leqsuc}
that $(t_{\leq\suc}\,u_2\,b;\s)\in\tvr{B_{\leq S(m)}}$ and hence 
$(\inr{t_{\leq\suc}\,u_2\,b};\s)\in\tvr{A(n)\lor B_{\leq S(m)}}$.
\end{proof}

\begin{cor}\label{r:t_phi}
Let $A(x)$ and $B(x)$ be any formulas whose only free variable is $x$. 
 Then
 \[t_\Phi\ureal\left(\fabvs{x}.\fabvs{y}.(A(x)\lor B(y) )\right) \imp \fabvs{x}.\fabvs{y}.\Phi_{A,B}(x,y),\]
 where ${t_\Phi }\defeq {\lambda hwx.\rec\,(t_0\,{h}\,w\,x)\,(t_\suc \,{h}\,w\,x)}$.
\end{cor}

\begin{proof}
Let $\s$ be a state, $h$ be a term such that $(h;\s) \real \fabvs{x}.\fabvs{y}.\left(A(x)\lor B(y)\right)$, 
and $n\in\N$ be a natural number.
We want to show that $(t_\Phi\,h\,\anyt\,\bar p;\s) \real \fabvs{y}.\Phi_{A,B}(p,y)$.
By definition of $t_\Phi$, we have 
\[t_\Phi\,h\,\anyt\,\bar n ~\reds \s \s~ \rec\, (t_0\,h\,\anyt\,\bar n)\,(t_\suc\,h\,\anyt\,\bar n).\]
Then, the result follows directly from External Induction (\autoref{r:external_induction})
and Lemmas \ref{r:t_0} and \ref{r:t_suc}.
\end{proof}

We can now take advantage of these terms to define the terms necessary 
to realize the formula $A_{\leq x}\lor B_{\leq x}$ where the role of $A$ and $B$ 
is now made symmetric again, using a second External Induction.
 
\begin{lem}\label{r:t_Delta}
For any internal formulas $A$ and $B$, we have:
 \[t_\Delta \ureal \left(\fabvs{x}.\fabvs{y}. (A(x) \vee B(y))\right) \to \fabvs{x}.(A_{\leq x} \to A_{\leq S(x)}\lor B_{\leq S(x)}),\]
 where 
 ${  t_{\Delta} }\defeq {\lambda {h}wxv.  \mww{c}{\left(t_\Phi\,{h}\,\,\anyt\,(\suc x)\,\anyt\,(\suc x)\right)}{\inl{t_{\leq\suc}\,v\,c_1}}{\inr{c_2}}}$.

 \end{lem}
\begin{proof}

 Let $\s\in\S$ be a state, $n\in\N$ be a natural number
 and $h,v\in\Lambda$ two terms
 such that $(h;\s)\in\tvr{\fabvs{x}.\fabvs{y}.(A(x)\lor B(y))}$ and 
 $(a;\s) \in\tvr{A_{\leq n}}$.
 By construction, we have
 \[
  t_{\Delta}\,\anyt\,\bar n\,a ~\reds \s \s ~
  \mww{c}{\left(t_\Phi\,{h}\,\,\anyt\,(\suc n)\,\anyt\,(\suc n)\right)}
  {\inl{t_{\leq\suc}\,a\,c_1}}{\inr{c_2}}.
 \]
To conclude by anti-reduction, we need to show that the reduced term
realizes $A_{\leq S(n)}\lor B_{\leq S(n)}$.
Using \autoref{r:t_phi}, we get that 
\[(t_\Phi\,{h}\,\,\anyt\,(\suc n)\,\anyt\,(\suc n);\s)
\real A(S(n))\lor B_{\leq S(n)}.\]
Using the adequacy of the $\orerule$ rule (for which the $\inr{\cdot}$ case is immediate),
we now have to prove that for any term $v$ such that $(v;\s)\real A(S(n))$,
$t_{\leq\suc}\,a\,v;\s \real A_{\leq S(n)}$. This follows from
\autoref{r:t_leqsuc} and the assumptions on $a$ and $v$.
\end{proof}

\begin{lem}\label{r:u0}
For any internal formulas $A$ and $B$, we have 
 \[u_0\ureal \left(\fabvs{x}.\fabvs{y}.(A(x)\lor B(y) )\right) \to \big (A_{\leq 0}\lor B_{\leq 0}\big),\]
 where $u_0 \defeq  \lambda h.\matchwith {(h\,\anyt\,0\,\anyt\,0)} a {\inl{t_{\leq0}\,a}}
 b {\inl{t_{\leq0}\,b}}$.
 \end{lem}
 \begin{proof}
 Let $\s\in\S$ be a state and $h\in\Lambda$ a term such that
 $(h;\s)\Vdash \fabvs{x}\fabvs{y}.(A(x)\lor B(y) )$.
 By construction, we have
\[
 u_0 \, h
 ~\reds \s \s  ~
 \matchwith {(h\,\anyt\,0\,\anyt\,0)} a {\inl{t_{\leq0}\,a}}
 b {\inl{t_{\leq0}\,b}}.
\]
The result easily follows by anti-reduction, using the adequacy of the $\orerule$ rule 
and \autoref{r:t_leq0}.
 \end{proof}

 \begin{lem}\label{r:usuc}
 For any internal formulas $A$ and $B$, we have
 \[
  u_\suc \ureal \left(\fabvs{x}.\fabvs{y}.(A(x)\lor B(y) )\right) \to
    \fabvs{x}.\left(\left(A_{\leq x}\lor B_{\leq x}\right)
\to\left(A_{\leq S(x)}\lor B_{\leq S(x)}\right)\right),
\]
where ${ u_\suc }\defeq { \lambda hwxd.\mww{d}{d}{t_\Delta\,{h}\,w\,x\,d_1}{t_{\Delta}\,(t_\lor\,{h})\,w\,x\,d_2}}$.
  \end{lem}
   \begin{proof}
  Let $\s\in\S$ be a state, $n\in\N$ a natural number
  and $h,v\in\Lambda$ be two terms such that 
  $(h;\s)\real \fabvs{x}.\fabvs{y}.(A(x)\lor B(y) )$ and $(v;\s)\real A_{\leq n}\lor B_{\leq n}$.
  By construction, we have
  \[
   u_\suc \, h \, \anyt \, \bar n\,v
   ~\reds \s \s ~
   \matchwith {v} a {t_\Delta\,{h}\,\anyt\,\bar n\,a} b {t_{\Delta}\,(t_\lor\,{h})\,\anyt\,\bar n\,b}.
  \]
  To conclude by anti-reduction, we use the adequacy of the $\orerule$ rule to prove that
  the reduced term belongs to $\tvr{A_{\leq S(n)}\lor B_{\leq S(n)}}$.
  For the $\inl \cdot$ case, if $w$ is a term such that $(w;\s)\real A_{\leq S(n)}$,
  then $(t_\Delta\,{h}\,\anyt\,\bar n\,w;\s)\real A_{\leq S(n)}\lor B_{\leq S(n)}$
  by \autoref{r:t_Delta} as expected.
  The $\inr \cdot$ case is symmetric, using $t_\lor$ and \autoref{r:t_lor}.
 \end{proof}

\begin{cor}\label{r:t_aux}
For any internal formulas $A$ and $B$, we have
 \[t_{\mathrm{aux}}\ureal \left(\fabvs{x}.\fabvs{y}.(A(x)\lor B(y) )\right) \to \fabvs{x}.\left (A_{\leq x}\lor B_{\leq x})\right),\]
 where  $t_{\mathrm{aux}}\defeq{ \lambda h.\rec\, (u_0\,h)\, (u_S\,h)}$.  
 \end{cor}
 \begin{proof}
  
 Let $\s$ be a state, and $(h;\s)\real \fabvs{x}\fabvs{y}(A(x)\lor B(y)$.
 We have that
$$t_{\mathrm{aux}} \,h\reds \s \s \rec\, (u_0\,h)\, (u_\suc\,h),$$
 so the result easily follows from \autoref{r:external_induction} and Lemmas \ref{r:u0} and \ref{r:usuc}.
 \end{proof}

 We are now ready to prove the main theorem of this section, by combining all 
 the terms into a realizer of LLPO$^\mathrm{st}$.
\begin{thm}[LLPO$^\mathrm{st}$]\label{r:llpo} 
Let $A$ and $B$ be internal formulas, we have
\[t_{\mathrm{LLPO}} \ureal \fabvs{x}.\fabvs{y}. (A(x) \vee B(y)) \to (\fabvs{x}.A(x) \lor \fabvs{y}.B(y)),\] 
where $t_{\mathrm{LLPO}} \defeq \lambda h. \mw{(t_{\mathrm{aux}}\, h\,\anyt\,\get)}{\inl{t_\delta\,x_1}}{\inl{t_\delta\,x_2}}$.
\end{thm}

\begin{proof}
 For any natural number $n\in\N$, 
if $(u;\s)\real \fabvs{x}\fabvs{y}(A(x)\lor B(y))$
using \autoref{r:t_aux} we get
 \[(t_{\mathrm{aux}}\,u;\s) \real\fabvs{x}.( \fabv{z}.z\leq x \to A(z) )\lor( \fabv{z}.z\leq x \to B(z)\big). \]
By Transfer (\autoref {r:transfer}), we then obtain 
\[(t_{\mathrm{aux}}\,u\,\anyt;\s) \real\fabv{x}.( \fabv{z}.z\leq x \to A(z) )\lor( \fabv{z}.z\leq x \to B(z)\big), \]
and in particular,
\[(t_{\mathrm{aux}}\,u\,\anyt;\s) \real\bvto{\nat \delta }( \fabv{z}.z\leq \delta \to A(z) )\lor( \fabv{z}.z\leq \delta \to B(z)\big). \]
Using \autoref{r:delta}, we get
\[(t_{\mathrm{aux}}\,u\,\anyt\,\get;\s) \real( \fabv{z}.z\leq \delta \to A(z) )\lor( \fabv{z}.z\leq \delta \to B(z)\big). \]
Using \autoref{r:t_delta} we indeed get a realizer of $\fabvs{z}. A(z)\lor\fabvs{z}.B(z)$.
\end{proof}

\section{A tainted quotient}
\label{s:some_name}
\newcommand{\tvq}[1]{\tv{#1}^*}
\newcommand{\satq}{\sat^*}
\newcommand{\redU}{\downarrow^\U}

In this section we explore the possibility of extending the work done above through a quotient, and the limitations of such construction. 
In \Cref{s:quotient}, we explain how this quotient can be obtained in a way that maintains the analogy with the Lightstone-Robinson construction.
The resulting theory is indeed an extension in which universal realizers 
for closed formulas are preserved and more principles are now realizable (\emph{e.g.} \autoref{r:st_lower}).
This makes it an even more convincing approach to nonstandard analysis from the point
of the captured theory, but not in terms of realizability. Indeed, as we will see in \Cref{s:limitations}, the terms witnessing the validity of formulas in the quotient 
can no longer be composed. On the other hand, as explained in \autoref{rmk:strict}, if one tries to be more faithful to the spirit of realizability, then the connection with nonstandard analysis is less convincing as one loses compatibility with \Los' theorem. Furthermore, the limitations do not seem to depend on the particular way one defines the quotient, as discussed in \Cref{s:food}.

\subsection{Realizability up to an ultrafilter}
\label{s:quotient} 
In order to fully mimic Lightstone and Robinson's construction, an extra step is required
where one takes a quotient of the interpretation with slices.
This step allows us to consider a more flexible notion of realizability 
where realizers are only required to be compatible with \emph{almost all} states,
in the sense that the set of compatible states belongs to the ultrafilter.

In order to simplify the discussion, and similarly to what was done in most of the paper, we don't include disjunction as a primitive connective.

Let us fix a free ultrafilter $\U$ over the set of states.
Given any set $V$, we denote by $\cong$ the 
equivalence relation over $V^\S$ defined by 
$f \cong g  \defeq \{\s\in \S:f(\s)=g(\s)\}\in\U$.

First, we can, within the realizability with slices,
change the way $\st f$ is interpreted 
to consider standardness up to the ultrafilter. In this way,
 $f\in\N^\S$ is said to be standard if and only if there exists $n\in\N$ s.t. $f \cong n^\inj$.
This allows to show, for instance, that nonstandard natural numbers are larger than standard ones.
\begin{prop}\label{r:st_lower}
 $\lambda x y.\loopp \ureal \forall x,y.\neg \stto x {\stto y {y<x}}$
\end{prop}
\begin{proof}
If $f\in\N^\S$ is a nonstandard individual 
 and $n\in\N$ any natural number, one proves by contradiction that
 $S=\{\s\in\S: n<f(\s)\}\in\U$.
 Indeed, otherwise one would have $\bar S\in\U$.
 
 For any $k \in \N$, let us write $S_k$ for the set $\{\s\in\S : f(\s) = k\}$.
 Since the sets  $S_0,...,S_n$ form a partition of $\bar S$, 
 it is easy to see that (exactly) one of these sets, say $S_m$, belongs to $\U$.
 Then $f \cong m^\inj$, which contradicts the fact that $f$ is nonstandard.

In particular, for any individuals $f,g$, any state $\s$, and any terms $t,u$ such that $(t;\s)\in\tvS{\stto f \bot} $
and $(u;\s)\in\tvS{\st g}$, we have that $f$ is necessarily nonstandard and that there exists 
$n\in\N$ such that $g\cong n^*$. By the claim above, we know
that there exists $\s'>\s$ such that $\s<f(\s')$. The result then follows by anti-reduction
from the fact that $\loopp \reds \s {\s'} \loopp$.
\end{proof}

We then need to define a new notion of realizability
in which realizers are also considered up to the equivalence relations
induced by $\U$.
To that end, a natural attempt consists in considering \Los' theorem as a guideline.
For the sake of clarity, let us denote by $\tvq{A}$ the truth values 
in this interpretation, which we shall call \emph{realizability up to $\U$}.

\begin{defi}
We say that a formula $A$ is \emph{\Los-reducible} 
if for any valuation $\rho$ closing $A$,
$t\in\tvq{A}$ if and only if $\{\s\in\S:(t;\s)\in\tvr A\}\in\U$.
\end{defi}

We actually define the interpretation of connectives 
by this equivalence. For example, the interpretation $\tvrq{A\to B}$
for the implication is defined by 
\[\{t\in\Lambda:\{\s\in\S:(t;\s)\in\tvr{A\to B}\}\in\U\},\]
while the interpretation of the quantifiers is still defined via intersections (resp. unions) 
over the same domain as in the interpretation with slices 
(\emph{e.g.}, $\tvrq{\forall x.A} \defeq \bigcap_{f\in\N^\S} \tvq{A}_{\rho,x\mapsto f}$).

\newcommand{\realq}{\Vdash^{\!\!\raisebox{.2pt}{\scriptsize *}}}
\begin{defi}[Realizability up to $\U$]\label{def:realU}
The interpretation of a formula $A$ together with a valuation $\rho$
 closing $A$ is the set $\tvrq{A}$ defined inductively according to the following clauses:
\[ 
\begin{array}{rcl}
 \tvrq{\st{f}}         & \defeq & \left\{\begin{array}{ll} \Lambda & \text{ if $f\cong n^\inj$, for some $n\in\N$}\\ \emptyset & \text{otherwise}\end{array}\right.\\[2mm]
  \tvrq{X(e_1,\dots,e_n)}& \defeq & 
                                    \{t \in \Lambda: \{\s\in\S:(t;\s)\in\rho(X)@(\foint{e_1}_{\rho},\dots,\foint{e_n}_{\rho})\}\in\U\}\\[2mm]
 
 \tvrq{\bvto {\nat e} A  }     & \defeq & \{t\in\Lambda : \{\s\in\S:(t;\s)\in\tvr{\bvnat e A}\}\in\U\}\} \\[2mm]
 \tvrq{A\imp B         }& \defeq & \{t\in\Lambda : \{\s\in\S:(t;\s)\in\tvr{A\to B}\}\in\U\} \\ 

 \tvrq{A_1\land A_2        }& \defeq & \{t\in\Lambda : \{\s\in\S: (\pi_1(t);\s)\in\tvS{A_1}_\rho\land (\pi_2(t);\s)\in\tvS{A_2}_\rho\}\in\U\}  \\[2mm]

 \tvrq{\forall x.A     }& \defeq & \bigcap_{f\in\N^\S}\tvq{A}_{\rho,x\gets f} \\[2mm] 
 \tvrq{\exists x.A     }& \defeq & \bigcup_{f\in\N^\S}\tvq{A}_{\rho,x\gets f} \\[2mm] 
 \tvrq{\forall X.A     }& \defeq & \bigcap_{F:\N^k\to\sat}\tvq{A}_{\rho,X\gets F} \\[2mm] 
 \tvrq{\exists X.A     }& \defeq & \bigcup_{F:\N^k\to\sat}\tvq{A}_{\rho,X\gets F} 
 \end{array}
\] 
We write $t\realq A$ if $t\in\tvq{A}$.
\end{defi}

As shown in the following theorem, first-order quantifiers behave well w.r.t.\ the ultrafilter.
\begin{thm}[\Los' theorem]\label{r:los}
First-order internal formulas as well as arbitrary conjunctions and implications
are \Los-reducible.
\end{thm}
\begin{proof}
\input{los}
\end{proof}

\autoref{r:los} implies that if a term $t$ is a realizer 
of a first-order internal formula $A$ ``often enough'' in the interpretation with slices, 
then $t$ is still a realizer in the interpretation up to $\U$. 
Since all the realizers in \Cref{s:nonstandard} were universal, 
 they are still  realizers in this new setting,
meaning that all the results from that section remain valid
in the interpretation up to $\U$.
In particular, \autoref{r:los} applies to Transfer, Idealization, 
Overspill or Underspill.

A simple example illustrating this new interpretation is the formula $\fas x. x<\delta$, 
which was realized by $\loopp$ in the interpretation with slices (see \autoref{rmk:diag})
and is now realized by any term (because for any $n\in\N$,
the set of states such that $n<\delta$ is equal to $[n;+\infty[$
which belongs to $\U$). 
Similarly, $\loopp$ can be replaced by $\anyt$ in \autoref{r:st_lower}. 
More generally, such a quotient allows us to get realizers for principles 
that were inaccessible in the interpretation with slices 
(e.g., $\forall x,y.\neg \stto x {\stto y {y<x}}$) but are usually valid
in nonstandard interpretations.

A more involved example concerns the Standardization principle:
 \emph{prima facie} this principle does not seem to be realizable with the current definitions,
but it is available for \Los-reducible formulas that are specified by the previous theorem. 
Technically, to internalize this in our interpretation would require to go to higher-order logic
in order to refer to standard predicates in the syntax.
In fact, this should not be a problem, as shown by the fact 
that our interpretation induces an evidenced frame~(see \Cref{s:ef}), which
is known in turn to induce triposes, hence a model of higher-order logic.

For the sake of simplicity, we will just consider standard predicates through
their semantical characterization, that is the ones induced by predicates 
in \autoref{def:realizability}, \emph{i.e.} whose value is identical in each slice.
Assume that we are given a \Los-reducible formula $A(x)$, which we identify 
with the truth-value function it induces $\mathbb{A}:n\in\N \mapsto \tvr{A(n)}$,
together with a standard predicate $\X$ (which is then \Los-reducible as well).
Restricted to this setting, Standardization states that there exists a standard 
predicate $\mathbb{Y}$ such that for any standard natural number $n\in\N$, we have:
\[\exists t. t \realq \mathbb{Y}(n^*) ~~\Leftrightarrow~~ 
  \exists t. t \realq \X(n^*) \land \mathbb{A}(n^*)\]
Since $\mathbb A$, $\X$ and arbitrary conjunctions are \Los-reducible, it is enough to 
consider the following standard predicate (seen via the function from $\N$ to $\Lambda$
inducing its actual value in the interpretation with slices):
\[\mathbb Y(n) \defeq \{t \in\Lambda : \{\s\in\S: t\real \X(n^*)\land \mathbb A(n^*)\}\in\U\}\]

Before taking a quotient, the definitions in \Cref{s:stateful_real}
gave us access to Standardization for internal formulas.
In the current setting, the restriction to \Los-reducible formulas is 
necessary to make it work with the quotient.
This is to be compared with Lightstone-Robinson's construction, where all internal formulas are \Los-reducible,
an analogous definition gives access to Standardization for these formulas.
Nonetheless, to validate the unrestricted principle of Standardization 
(that is, where $A$ is any formula), one usually needs to use a more involved
construction of a model by means of an adequate ultralimit, and the proof that
Standardization holds relies on transfinite reasoning and the full Axiom of Choice.
None of these principles being computationally interpretable in our setting,
it seems that the restricted statement above is the best we can do here.

\subsection{Limitations of the construction}
\label{s:limitations}
While the quotient from the previous subsection allows to capture a theory
which is even closer to ``usual'' nonstandard analysis, 
it has some drawbacks with respect to its realizability facet.
The main drawback, which is highlighted in \autoref{rmk:strict}, concerns the interpretation of implication.
First, let us connect the interpretation of implication through the quotient
with the expected interpretation in usual realizability settings.

\begin{prop}\label{r:quotient_imp}
For any internal formulas $A$ and $B$, 
and any valuation $\rho$ closing both $A$ and $B$, we have 
 $\tvrq{A\imp B} \subseteq \{t:\forall u\in\tvrq{A}.t\,u\in\tvrq{B}\}$.
 \end{prop}

 \begin{proof}
 {
For any term $t$ and any formula $A$, 
let us denote by $S^{A}_{t}$ the set $\{\s\in\S:(t;\s)\in\tvr{A}\}$.
Let $t\in\Lambda$ be such that $S^{A\imp B}_t\in\U$ and $u\in\tvrq{A}$.
By hypothesis, $S^A_u\in\U$.
We need to show that $tu \in \tvrq{B}$. 
Again, for any $\s\in S^{A\imp B}_t \cap S^A_u \in \U$, we have $tu;\s \in \tvr{B}$.
By upwards closure, we deduce that $\{\s: (tu;\s)\in\tvr{B}\}\in\U$,
hence $tu\in \tvrq{B}$, and the result follows from \autoref{r:los}.
 }
{By \autoref{r:los}, we have
\[\tvrq{A}= \{ u\in\Lambda:\{\s\in\S: (u;\s)\in \tvr{A}\}\in\U\ \} \quad \text{and}\]
\[\tvrq{B}= \{ v\in\Lambda:\{\s\in\S: (v;\s)\in \tvr{B}\}\in\U\ \}.\] 
For any term $t$ and any formula $A$, 
let us denote by $S^{A}_{t}$ the set $\{\s\in\S:(t;\s)\in\tvr{A}\}$.
 
Let $t\in\Lambda$ be such that $S^{A\imp B}_t\in\U$ and $u\in\tvrq{A}$.
By hypothesis, $S^A_u\in\U$.
We need to show that $tu \in \tvrq{B}$. 
Again, for any $\s\in S^{A\imp B}_t \cap S^A_u \in \U$, we have $tu;\s \in \tvr{B}$.
By upwards closure, we deduce that $\{\s: (tu;\s)\in\tvr{B}\}\in\U$,
hence $tu\in \tvrq{B}$, and the result follows.}
\end{proof}
\begin{rem}\label{rmk:strict} 
One could have been tempted to define the truth value $\tvrq{A\to B}$
as the set of terms $t$ such that for any $u\in\tvrq{A}$, $t\,u\in\tvrq{B}$,
as is usual in realizability. 
Unfortunately, such a definition is incompatible with \autoref{r:los},
 as the other inclusion in \autoref{r:quotient_imp} does not hold.
To see this, let $A\defeq\natp{\tau}$ and $B\defeq\bot$ 
where $\tau$ is a non-computable function\footnote{To that end,
one can for instance consider the function $\tau$ which to each $\s\in\S$ associates 
the smallest natural number $n\in\N$ such that there is no term of size smaller than or equal to $\s$
that computes $n$ the state $\s$:
$\tau(\s)\defeq \inf\{n\in\N:\neg \exists t. |t|\leq \s \land t \reds {\s} {\s} \overline{n}\}$ .}
$\tau:\S\to\N$ for which
  there is no term $u$ such that $\forall \s. u\reds{\s}{\s}\tau(\s)$.
  By construction, we have that $\tvq{A}= \emptyset$, so that
   obviously for any $ u\in\tvq{\natp{\tau}}$, the function $(\lambda x.x)\,u\in\tvq{\bot}$.
  Yet,  for each state $\s$ the truth value $\tvr{\natp{\tau}}$ is not empty
  (it contains at least $(\overline{n},\s)$, for $n=\tau(\s)$) 
  and  therefore $(\lambda x.x;\s)\notin\tvr{\natp{\tau}\to \bot}$
  (since for any $(u;\s)\in\tvS{\natp{\tau}}$, $((\lambda x.x)\,u;\s)\notin\tvq{\bot}$).
\end{rem}

As it turns out, \autoref{def:realU} is not as compositional as one would expect
in realizability. 
Indeed, we can compose a realizer $t\in\tvrq{A\to B}$
with a realizer in $u\in\tvrq{A}$ to get $t\,u\in\tvrq{B}$,
but the \impirule-rule is not adequate when considering substitutions of
variables by realizers in the quotiented truth values.
In particular, \autoref{rmk:strict} emphasizes that the structure of this interpretation 
does not induce an evidenced frame, since it is not possible to define 
the function $\lambda:E\to E$ necessary to interpret implication.

\subsection{Stranger things}
\label{s:food}

As mentioned above, \autoref{rmk:strict} highlights  the existence 
  of ``counter-intuitive'' peculiarities of the interpretation up to $\U$ with respect 
  to the quotient in the Lightstone-Robinson construction.
 The latter indeed appears to be more regular, seemingly for two main reasons.
 
 First, while the Lightstone-Robinson construction is based on Boolean-valued models, 
 realizability interpretations associate to each formula a set of realizers 
 (instead of one unique Boolean). Besides, the use of relativized quantifiers (for instance
 in the statement for Idealization) forces us to use only computable 
 functions\footnote{This is the reason why, for instance, the premise of Idealization needs to be restricted to the existence of a \emph{standard} natural number $x$, instead of any natural number as is usually the case.}.
 
 Second, as highlighted in \Cref{s:slices}, in the stateful interpretation
 the $\set$ instruction allows terms to change the value of the states during computations,
 and thus of the slices. This phenomenon does not occur in the Lightstone-Robinson  construction where slices of the product are completely separated between them. 
 In fact, the ability of reading the value of the slice already implies that
  propositional internal formulas do not induce standard truth values, which is counter-intuitive.
  \begin{prop}
  There exists an internal propositional formula $A$, a term $t$ and two states $\s_0,\s_1$
  such that $(t;\s_0)\in\tvS{A}$ but $(t;\s_1)\notin\tvS{A}$.
 \end{prop}
 \begin{proof}
  Take for instance $A\defeq (A_1\land \neg A_1)\to A_1$,
  $t \defeq \rec\,(\lambda x.\pi_1(x))\,(\lambda xyz.\pi_2(z))\,\get$,
  $\s_0=0$ and $\s_1=1$.
  We have:
  \begin{itemize}
   \item $(t;0)\in\tvS{A}$: 
   for any $(u;0)\in\tvS{A_1\land \neg A_1}$, $u\reds 0 \s (u_1,u_2)$ with $(u_1;\s)\in\tvS{A_1}$,
   and
   \[t ~~\reds 0 0~~ (\rec\,(\lambda x.\pi_1(x))\,(\lambda xyz.\pi_2(z))\,0) \,u 
  ~~ \reds 0 0 ~~(\lambda x.\pi_1(x))\,u ~~\reds 0 0 ~~\pi_1(u) ~~\reds 0 \s ~~u_1.\]
   The result follows by anti-reduction.
   \item $(t;1)\notin\tvS{A}$: 
   for any $(u;1)\in\tvS{A_1\land A_2}$, $u\reds 1 \s (u_1,u_2)$ with $(u_2;\s)\in\tvS{\neg A_1}$,
   and
   \[t ~~\reds 1 1~~ (\rec\,(\lambda x.\pi_1(x))\,\lambda xyz.\pi_2(z)\,1) \,u 
  ~~ \reds 1 1 ~~(\lambda xyz.\pi_2(z)) 0\,(\rec ...)\,u ~~\reds 1 1 ~~\pi_2(u) ~~\reds 1 \s ~~u_2.\]
   Since $(u_2;1)\in\tvS{\neg A_1}$, it cannot be the case that $(u_2;1)\in\tvS{A_1}$. \qedhere
  \end{itemize}
 \end{proof}

 As explained above, the interpretation briefly introduced in \Cref{s:quotient} is an attempt
 to provide a quotient, guided by the rationale of \Los' theorem. 
 Nonetheless, there might be more refined ways to arrive at a satisfying definition 
 of a quotient. In particular, \autoref{def:realU} does not take the computation into account when defining the quotient, which turns out to be problematic as \autoref{rmk:strict} shows.
 Therefore, it could be tempting to contemplate a notion of reduction up to $\U$ as follows
 \[ t\redU u ~~\defeq~~ \{\s\in\S : \exists \s'.t\reds \s {\s'} u\}\in \U \]
 Nonetheless, such a definition leaves us even further form our goal since the induced realizability interpretation would present a lot of counter-intuitive peculiarities.
 We illustrate some of these peculiarities in the following propositions,
 for which we consider a free ultrafilter $\U$ on $\S$ and, without loss of generality, assume that 
    $\{\s\in\S:\exists n\in\N. \s=2n\}\in \U$.
 In particular, this implies that $\{\s\in\S:\exists n\in\N. \s=2n+1\}\notin \U$.
 We will say that a property occurs \emph{often enough} when the set of states for which
 it is valid belongs to $\U$.

 The next proposition shows the existence of a formula $A$ and two terms $t$ and $u$ such that 
 $u$ is often enough a realizer and $t$ reduces often enough to $u$, 
 but never to a slice in which it is a realizer.
   \begin{prop}
  There exist a formula $A$, a valuation $\rho$ closing $A$ and two terms $t$, $u$ such that
  \begin{enumerate}
   \item $\{\s\in\S:(u;\s)\in\tvr{A}\}\in\U$
   \item $\{\s\in\S:\exists \s'. t \reds \s {\s'} u\}\in\U$
   \item $\forall \s, \s'. t \reds \s {\s'} u \Rightarrow (u;\s')\notin\tvr{A}$.
  \end{enumerate}
 \end{prop}
 \begin{proof}
Consider the (nonstandard) individual $\tau$ defined by $\tau : n\in\N \mapsto n \mod 2$
 (i.e $\tau(2n)=0$ and $\tau(2n+1)=1$). 
 By construction, we have
 \[\tvS{\tau =0^\inj}=\{(t;\s)\in\Lambda\times\S: \exists n\in N. \s=2n\}.\]
 
Let us now define a function $f$ which, given any integer $n\in\N$, 
 computes the lowest odd number greater than or equal to $n$:
 $f(0) = 1$, $f(1) = 1$, $f(2)=3$, etc.
 It is clear that this function is primitive recursive, hence
 there is a term $\mathsf{next\_odd}$ that computes it.
We let $u\defeq\lambda x.x$ and 
$t \defeq\set\, (\mathsf{next\_odd} \get)\, u$.
For any state $\s\in\S$, we then have
\[t~=~\set \,(\mathsf{next\_odd} \get)\, u ~~\reds{\s}{\s}~~ \set\, (\mathsf{next\_odd}\, \s)\,
~~\reds{\s}{\s}~~\set \,\overline{f(\s)}\, u ~~\reds{\s}{f(\s)}~~ u,\]
where $f(\s)$ is odd. Hence, if we define $A\defeq x = 0^*$ and $\rho = x\mapsto \tau$, 
we have
\begin{enumerate}
 \item $\{\s\in\S:(u;\s)\in\tvS{\tau=0^\inj}\}=\{\s\in\S:\exists n\in\N,\s=2n\}\in\U$
 \item $\{\s\in\S:\exists \s'. t \reds \s {\s'} u\} = \S \in\U$
 \item for any $\s$, $t \reds \s {f(\s)} u$ and $(u;f(\s))\notin\tvS{A}$ since $f(\s)$ is odd.\qedhere
\end{enumerate}
\end{proof}

The next result shows that even if there are enough slices in which $t$ reduces to $u$ 
in a slice that makes it a realizer of some formula $A$,
$u$ may not be a realizer of $A$ often enough.
\begin{prop}
  There exist an atomic formula $A$, a valuation $\rho$ closing $A$ and two terms $t$, $u$ such that
  \begin{enumerate}
   \item $\{\s\in\S:\exists \s'. t \reds \s {\s'} u \land (u;\s')\in\tvr{A}\}\in\U$
   \item $\{\s\in\S:(u;\s)\in\tvr{A}\}\notin\U$.
  \end{enumerate}
 \end{prop}
\begin{proof}
 Take again the (nonstandard) individual $\tau$ defined by $\tau : n\in\N \mapsto n \mod 2 $,
 $\rho \defeq x \mapsto \tau$ and $A\defeq x=1^*$.
Let us define $u\defeq \lambda x.x$, $\incr\defeq \lambda x.\set\,(\suc\,\get)\,x$
$t\defeq \incr\,u$.
By construction, we have that $\tvS{\tau = 1}=\{(v;\s):\exists n\in\N, \s=2n+1\}$
and $t=\set \,(\get + 1)\, u \reds{\s}{\s+1} u$. Hence
\begin{enumerate}
 \item $\{\s\in\S:\exists \s'. t \reds \s {\s'} u \land (u;\s')\in\tvr{A}\}
 =\{\s\in\S:\exists n. \s=2n\}\in\U$
 \item $\{\s\in\S:(u;\s)\in\tvr{A}\}=\{\s\in\S:\exists n. \s=2n+1\}\notin\U$.\qedhere
\end{enumerate}
\end{proof}

\section{Related and future work}
\label{s:conclusion} 

\subsection{Related work}\label{s:related_work}

Some related works concern notions of realizability for nonstandard arithmetic
which are variants of Kreisel's modified realizability \cite{BerBriSaf12,DinGas18}. 
These notions of realizability are more inspired by Nelson's syntactical approach 
to nonstandard analysis. In particular, they rely on translations of formulas inducing
conservative extensions of Heyting arithmetic. 
To draw a comparison with Van den Berg \emph{et al.}'s work, it should be observed 
that they interpret standard elements as finite sequences that can be thought 
of as a process of accumulating potential witnesses. 
In particular, their interpretation
crucially relies on a monotonicity property for these sequences (regarding sequence inclusion),
stating that realizers are provably upwards-closed~\cite[Lemma 5.4]{BerBriSaf12}. 
This property has no counterpart in our setting.
On the other hand, our interpretation is able to give computational content to
nonstandard individuals, and even to give explicit nonstandard elements (such as the diagonal)
with their corresponding realizers.
This is, for example, what allows us to computationally interpret 
Idealization (see \autoref{r:idealization}),
whereas the functional interpretation  for Idealization in~\cite{BerBriSaf12}
is trivial in the sense that the interpretations of the premise and the conclusion 
of any instance of Idealization are identical. 
It could be interesting to better understand the relation between this approach and the approaches based on Kreisel's realizability. In particular, we would like to know whether we can obtain
a preservation result for some class of formulas (\emph{e.g.} internal, quantifier-free, $\exists$-free 
formulas).

Similar ideas have been addressed by Aschieri. In~\cite{Aschieri17} the author
uses a notion of state which allows to construct a forcing model. In particular,
natural numbers are interpreted as functions from states to $\N$. 
Yet, his work does not pay attention to the nonstandard principles that can be obtained in his
setting but rather to forcing.
It would be natural to investigate whether our setting also allows for forcing 
techniques.
This connection with forcing is reinforced by the fact that
in the realm of Krivine's realizability, which generalizes
Cohen's forcing, the latter is given a computational content
via the addition of a monotone memory cell to the abstract machine
in order to store forcing conditions \cite{Krivine11,Miquel11b}.
Also, recent work of Powell has been focusing on a variant of Gödel's functional
interpretation to take into account stateful computations~\cite{Powell18}.
In addition to investigating the computational contents of the stateful programs
obtained by extraction through this interpretation, the author
proposes some problems that the reader might find interesting.

\subsection{Weak K{\H o}nig's Lemma}
As shown in \cite{DinFer16}, $\mathrm{WKL}_0$ 
(one of the Big Five systems from Reverse Mathematics) 
is interpretable, over a nonstandard version of primitive recursive 
arithmetic with extensionality, using a version of the Axiom of Choice and 
Idealization. It relies on distinguishing two sorts: 
the number sort is interpreted by the standard numbers,
and the set sort is interpreted by bounded type 1 functionals 
(or by number codes, both standard and nonstandard, of finite sets of numbers,
again both standard and nonstandard). 

Recall that Weak K{\H o}nig's Lemma states that every infinite binary tree has
an infinite branch.
As it turns out, in that context to say that $T$ is a tree is to say two things: 
$(i)$  every standard natural number which is in the tree is the code of a binary sequence and 
$(ii)$ if some standard $\sigma$ is in the tree and the binary sequence coded by a standard 
element $\tau$ is an initial segment of the binary sequence coded by $\sigma$,
then $\tau$ is also on the tree. 
The interpretation of being infinite is a formula saying that for every standard 
natural number $\sigma$ there exists a standard element with length $w$ which is 
in the tree.
The proof then relies on showing that an element $\alpha$, defined exactly as
$\sigma$ below and at $w$ and 0 from there onwards can be turned into an infinite 
branch with the use of Idealization.

So, the interpretation of Weak K{\H o}nig's Lemma crucially relies on the ability 
to manipulate trees and on Idealization. Of course, in our setting, we have an 
explicit (nontrivial) realizer for Idealization, so, in principle, it should be 
possible to give a realizer for Weak K{\H o}nig's Lemma. However, that would 
require a whole reformulation of the framework in order to have an explicit 
access to trees instead of a noncomputational second-order quantification.

\subsection{Horizons}
The work done in this paper raises some natural questions of which we mention 
a few, as possible lines of investigation.

A first natural question comes from the fact that prior interpretations 
of nonstandard arithmetic, such as \cite{BerBriSaf12,DinGas18, FerGas15} 
(and also \cite{FerOl05} and \cite{FerNunes06} in a context that does 
not involve nonstandard arithmetic), restrict quantifiers by bounding
the variables under their scope. It is then pertinent to ask whether 
this could be given a more computational interpretation as we do here, 
in order to see it as some kind of ``computation up to (the bound)''.

A second possible path would be to reformulate our interpretations in order 
to account for classical logic by using control operators as is usual in Krivine's 
realizability~\cite{Krivine09}. Alas, our attempts in that direction have 
not been very fruitful, mostly because Krivine's interpretation
crucially relies on an orthogonality relation between 
terms and evaluations contexts (which is reminiscent of 
the duality of computation in classical logic~\cite{CurHer00}). 
In terms of the ultrafilter, this would require some sort of
perfect balance to make the quotient compatible with this orthogonality 
relation which so far has eluded us. 
This is similar to the limitations pointed out in \Cref{s:limitations}.

Thirdly, there is a very active line of research in realizability concerning 
the interpretation of various choice principles. In particular, the use of states 
or memoization has proven to be useful for interpreting dependent choice (e.g.\
in \cite{BerBezCoq98}, \cite{Herbelin12} or \cite{CohFarTat19}, to name but a few)
or Double Negation Shift (DNS) (in \cite{Blot22}, Blot uses an ``update recursion''
mechanism to realize DNS). At the same time, DNS is also interesting in itself 
as a non-intuitionistic principle. This is particularly relevant since our setting
interprets (a version of) the LLPO principle, which means that we are somewhere
between intuitionistic and classical logic.
Furthermore, DNS is also known to be interpretable using bar recursion, which 
raises the question of knowing whether our interpretation could be compatible 
with such an operator.

Finally, we would like to mention that Brede and Herbelin's~\cite{BreHer21}
establishes a hierarchy of choice principles, relating in particular 
tree-based choices principles and their dual bar induction-based principles. 
Many of the principles they study are not attached to a precise computational 
content so far, and so it would be interesting to see if there exist specific interpretations 
that could capture exactly each of these principles, and, in particular,
``lower'' instances of their generalized dependent choice or generalized bar 
induction principles.

\bibliography{biblio}

\bibliographystyle{alphaurl}

\end{document}

%% file: macros.tex
\newcommand{\sat}{\text{{\bf SAT}}}
\newcommand{\imp}{\rightarrow}
\newcommand{\defeq}{\triangleq}
\newcommand{\automath}[1]{\relax\ifmmode{#1}\else{$#1$}\fi}

\makeatletter
\newcommand{\oset}[3][0ex]{%
  \mathrel{\mathop{#3}\limits^{
    \vbox to#1{\kern-2\ex@
    \hbox{#2}\vss}}}}
\makeatother

\newcommand{\red}{\rightarrow}

\newcommand{\V}{\mathcal{V}}

\newcommand{\prfcase}[1]{\paragraph*{\normalfont\quad \textbf{\prfcasetext} #1}}
\renewcommand{\prfcase}[1]{\paragraph*{\normalfont\quad \underline{#1}}}

\newcommand{\real}{\Vdash}
\newcommand{\ureal}{\Vvdash}

\newcommand{\loopp}{\mathsf{loop}^+}
\definecolor{myblue}{rgb}{0,0.5,0.8}

\newenvironment{mycenter}{%
  \setlength\topsep{1pt}
  \setlength\parskip{1pt}
  \begin{center}
}{%
  \end{center}
}

\renewcommand{\P}{\mathcal{P}}

\makeatletter
\newcommand{\uset}[3][0ex]{%
  \mathrel{\mathop{#3}\limits^{
    \vbox to #1{\kern-2\ex@
    \hbox{\tiny $#2$}\vss}}}}
\makeatother

\let\barredL\L
\newcommand{\Los}{{\barredL}o\'{s}\,}
\newcommand{\T}{\mathcal{T}}
\newcommand{\U}{\mathcal{U}}

\newcommand{\tval}[1]{|#1|}
\newcommand{\tv}[1]{\tval{#1}}

\newcommand{\M}{\mathcal{M}}
\newcommand{\F}{\mathcal{F}}

\newcommand{\FV}{{FV}}
\newcommand{\dom}{\mathop{\mathrm{dom}}}
\newcommand{\nomidem}{\vspace{-0.5em}}

\newcommand{\bvto}[1]{\{#1\}\mapsto}
\newcommand{\bvnat}[1]{\{\nat {#1}\}\mapsto}
\newcommand{\byv}{{\mathrm{v}}}
\newcommand{\canon}[1]{\overline{#1}}

\newcommand{\fa}[2]{\forall^{#1}{#2}}
\newcommand{\far}{\forall^{\N}\!}
\newcommand{\fas}[1]{\fa{\mathrm{st}}{#1}}
\newcommand{\fars}[1]{\fa{\canon{\mathrm{st}}}{#1}}
\newcommand{\fabvs}[1]{\fa{\{\!\mathrm{st}\!\}}\!{#1}}
\newcommand{\fabv}[1]{\forall^{\{\!\N\!\}}\!#1}

\newcommand{\ex}[2]{\exists^{#1}{#2}}
\newcommand{\exs}[1]{\ex{\mathrm{st}}{#1}}

\newcommand{\exbv}[1]{\exists^{\{\!\N\!\}}\!#1}
\newcommand{\exbvs}[1]{\ex{\{\!\mathrm{st}\!\}}\!{#1}}
\newcommand{\exr}{\exists^{\N}\!}

\newcommand{\TYP}[3]{#1~{\vdash}~#2~{:}~#3}

\def\Lam#1#2{\lambda#1\,.\,#2}

\newcommand{\N}{{\mathbb{N}}}

\newcommand{\llpost}{{LLPO$^{\mathrm{st}}$}}

\newcommand{\rec}{\mathop{\mathsf{rec}}}
\newcommand{\nat}[1]{\mathrm{Nat}(#1)}
\newcommand{\natp}[1]{\mathrm{Nat'}(#1)}
\newcommand{\tvr}[1]{\tv{#1}_\rho}
\newcommand{\foint}[1]{\llbracket #1 \rrbracket}
\renewcommand{\st}[1]{\mathrm{st}(#1)}
\newcommand{\stto}[1]{\st{#1}\to}

\newcommand{\rcase}{\noindent$\underline{\Rightarrow}\!\raisebox{.875pt}{$\rfloor$}~$}
\newcommand{\lcase}{\noindent$\underline{\Leftarrow }\!\raisebox{.875pt}{$\rfloor$}~$}

\newcommand{\lequiv}{\Leftrightarrow}

\newcommand{\get}{\mathop{\mathsf{get}}}
\newcommand{\set}{\mathop{\mathsf{set}}}
\newcommand{\matchop}{\mathop{\mathsf{case}}}
\newcommand{\inl}[1]{\iota_1(#1)}
\newcommand{\inr}[1]{\iota_2(#1)}
\newcommand{\matchwith}[5]{\matchop\, #1\, \{\inl {#2}\mapsto #3|\inr {#4}\mapsto #5\}}
\newcommand{\mw}[3]{\matchwith{#1}{x_1}{#2}{x_2}{#3}}

\newcommand{\mww}[4]{\matchwith{#2}{#1_1}{#3}{#1_2}{#4}}
\newcommand{\incr}{\mathop{\mathsf{incr}}}

\newcommand{\ifop}{\mathop{\mathsf{if}}}
\newcommand{\thenop}{\mathop{\mathsf{then}}}
\newcommand{\elseop}{\mathop{\mathsf{else}}}
\newcommand{\ifte}[4]{\ifop\,#1 = #2\,\thenop\,#3\,\elseop\,#4}

\newcommand{\limp}{\Rightarrow}
\newcommand{\X}{\mathbb{X}}

\newlength{\myl}
\newlength{\mysep}
\newcommand{\twoelt}[2]{
  \settowidth{\myl}{$#1$}
  \setlength{\mysep}{0.5\textwidth}
  \addtolength{\mysep}{-\myl}
  \addtolength{\mysep}{-15pt}
\[
  #1
  \hspace{\mysep}
  #2
\]
}

\newcommand{\CA}{\automath{\mathrm{CA}}}
\newcommand{\Elim}{\automath{\mathrm{Elim}}}

\usepackage[notext,not1,nomath]{stix} 
\DeclareFontEncoding{LS1}{}{}
\DeclareFontSubstitution{LS1}{stix}{m}{n}
\DeclareSymbolFont{symbols2}{LS1}{stixfrak}{m}{n}
\DeclareMathSymbol{\llangle}{\mathopen}{symbols2}{"28}
\DeclareMathSymbol{\rrangle}{\mathclose}{symbols2}{"29}

\newcommand{\xle}[1]{\xrightarrow{\!#1\!}}
\newcommand{\eid}{e_{\mathtt{id}}}

\newcommand{\ecomp}[2]{#1 \mathop{;} #2}

\newcommand{\etrue}{e_{\scriptscriptstyle\top}}

\newcommand{\epair}[2]{\llangle #1, #2 \rrangle}

\newcommand{\efst}{e_{\mathtt{fst}}}

\newcommand{\esnd}{e_{\mathtt{snd}}}

\newcommand{\elambda}[1]{\lambda #1}

\newcommand{\eeval}{e_{\mathtt{eval}}}

\newcommand{\power}{\mathcal{P}}
\newcommand{\efimp}{\mathop{\supset}}

%% file: rules.tex
\newcommand{\elimmark}{E}
\newcommand{\intromark}{I}

\newcommand{\autorule}[1]{\relax\ifmmode{\scriptstyle(#1)}\else$(#1)$\fi}

\newcommand{\axrule }{\autorule{\textsc{Ax}  }}

\newcommand{\impirule}{\autorule{\imp_\intromark}}
\newcommand{\imperule}{\autorule{\imp_\elimmark }}

\newcommand{\fiurule}{\autorule{\forall^1_\intromark}}
\newcommand{\feurule}{\autorule{\forall^1_{\!\elimmark }}}
\newcommand{\fidrule}{\autorule{\forall^2_\intromark}}
\newcommand{\fedrule}{\autorule{\forall^2_{\!\elimmark }}}

\newcommand{\exiurule}{\autorule{\exists^1_\intromark}}

\newcommand{\exidrule}{\autorule{\exists^2_\intromark}}

\newcommand{\andirule }{\autorule{\land_\intromark }}

\newcommand{\andeurule}{\autorule{\land^1_\elimmark}}
\newcommand{\andedrule}{\autorule{\land^2_\elimmark}}

\newcommand{\oriurule}{\autorule{\lor^1_\intromark}}
\newcommand{\oridrule}{\autorule{\lor^2_\intromark}}
\newcommand{\orerule }{\autorule{\lor_\elimmark   }}

\newcommand{\recrule}{\autorule{\rec}}

%% file: real_sat.tex
{By a straighforward induction on the structure of $A$. 
Observe for instance that the case $\st{f}$ follows from the definition and 
that the case $X(e_1,\dots,e_n)$ follows from the fact that, by definition, $\rho(X)$ takes values in $\sat$.\qedhere }

{By induction on the structure of $A$.
The case $\st{f}$ is clear from the definition and the case $X(e_1,\dots,e_n)$ follows from the fact that, by definition, $\rho(X)$ takes values in $\sat$.

\prfcase{$A\imp B         $}
Let $t$, $t'$ be two terms such that  $(t';\s\prim)\in\tvr{A\to B}$ 
and $\steps{\s}{\s\prim}{t}{t'}$ for some states $\s,\s\prim$.
Let  $(u;\s\prim)\in\tvr{A}$. We have that
$\steps{\s}{\s\prim}{t\,u}{t'\,u}$, which by definition belongs to $\tvr{B}$.
We conclude the result by the induction hypothesis for $B$. The same proof applies to the case $\bvto {\nat e} A$.

\prfcase{$A_1\land A_2         $}
Let $t$, $t'$ be two terms such that $(t';\s\prim)\in\tvr{A_1\land A_2}$ 
and $\steps{\s}{\s\prim}{t}{t'}$ for some states $\s,\s\prim$.
For any $i\in\{1,2\}$, we have that
$\steps{\s}{\s\prim}{\pi_i(t)}{\pi_i(t')}$, which by definition belongs to $\tvr{A_1}$.
We conclude the result by the induction hypothesis for $A_i$. 
The proof for the case $A_1 \lor A_2$ is analogous.

\prfcase{$\forall x.A     $}
Let $t$, $t'$ be two terms such that $(t';\s\prim)\in\tvr{\forall x. A}$ 
and $\steps{\s}{\s\prim}{t}{t'}$ for some state $\s,\s\prim$. 
By definition, for any $f\in\N^\S$, it holds that  $(t';\s\prim)\in\tvS{A}_{\rho,x\mapsto f}$.
Hence by the induction hypothesis for $A$, we get that $(t;\s)\in\tvS{A}_{\rho,x\mapsto f}$.
This being true for any $f\in\N^\S$, we deduce that $(t;\s)\in\tvr{\forall x.A}$.
The cases for the other quantifiers are similar. \qedhere
}

%% file: state_adequacy.tex
The proof, by case analysis, is essentially the same as the usual adequacy proof for HA2, 
since none of the instructions involved in the typing rules allows to  
change the value of the state.

In each case, we write $\Gamma$ for the typing context, $\rho$ for a valuation
closing all the considered formulas, $\s$ for the considered state 
and $\sigma$ for a substitution such that $(\sigma;\s)\real\rho(\Gamma)$.

%
%
%

\prfcase{\axrule  }
Directly from the assumption that $(\sigma;\s)\real\rho(\Gamma)$.

\prfcase{\impirule}
By assumption, for any substitution $\sigma'$ such that $(\sigma';\s)\real \rho(\Gamma),x:\rho(A)$,
we have that $(\sigma(t);\s)\in\tvr{B}$.
We have to prove that $(\lambda x.\sigma(t);\s)\in \tvr{A \to B}$.
Let then $u$ be a term such that $(u;\s)\in\tvr{A}$. By definition, we have
$\lambda x.\sigma(t)\,u ~\reds {\s} {\s}~\sigma(t)[u/x]$. 
Since $\sigma(t)[u/x]=(\sigma,x:=u)(t)$ and $(\sigma,x:=u;\s)\real \rho(\Gamma,x:A)$,
we obtain $(\sigma(t)[u/x];\s)\in\tvr{B}$. We conclude that
$(\lambda x.\sigma(t)\,u;\s)\in\tvr{B}$ by anti-reduction.

\prfcase{\imperule}
By assumption,
we have that $(\sigma(t);\s)\in\tvr{A\to B}$
and $(\sigma(u);\s)\in\tvr{A}$.
By the definition of $\tvr{A\to B}$, we obtain that $\sigma(t\,u)=\sigma(t)\,\sigma(u)\in\tvr{B}$.

\prfcase{\andirule}
By assumption,
 $(\sigma(t_1);\s)\in\tvr{A_1}$
and $(\sigma(t_2);\s)\in\tvr{A_2}$.
For any $i\in\{1,2\}$ we have that
$\pi_i(\sigma(t_1,t_2))=\pi_i(\sigma(t_1),\sigma(t_2)) \reds \s \s \sigma(t_i)$,
where the latter belongs to $\tvr{A_i}$. Then, by anti-reduction  $\pi_i(\sigma(t_1,t_2))\in\tvr{A_i}$
and hence $\sigma(t_1,t_2)\in\tvr{A_1\land A_2}$.

\prfcase{\andeurule}
By assumption, we have that $(t;\s)\in\tvr{A_1\land A_2}$,
which entails by definition that $(\pi_1(t);\s) \in\tvr{A_1}$.
The case  {\andedrule} is similar.

\prfcase{\exiurule}
By assumption, there exists $n\in\N$ such 
that $(t;\s)\in\tvr{A[x:=n]}=\tv{A}_{\rho,x\gets n^\inj}$. 
The result then follows from the fact that $\tv{A}_{\rho,x\gets n^\inj}\subseteq \bigcup_{f\in\N^\S}\tv{A}_{\rho,x\gets f}=\tvr{\exists x.A}$.

\prfcase{\feurule }
By assumption, $(t;\s)\in\bigcap_{f\in\N^\S}\tv{A}_{\rho,x\gets f}$. 
The result  follows directly from the fact that $\bigcap_{f\in\N^\S}\tv{A}_{\rho,x\gets f}\subseteq\tv{A}_{\rho,x\gets n^\inj} $.

\prfcase{\fiurule }
By assumption,  $(t;\s)\in\tv{A}_{\rho'}$ for any valuation $\rho'$ closing $A$
which, since $x$ does not occur in  $\Gamma$, can freely map $x$ to any individual in $\N^\S$.
In other words, $(t;\s)\in\bigcap_{f\in\N^\S}\tv{A}_{\rho,x\gets f}$.
The case for {\fidrule} is similar.

\prfcase{\autorule{\cong}}
Directly from \autoref{r:cong_states}. \qedhere

%% file: glueing.tex
The proof is by induction on the structure of $A$. 
{
\prfcase{ $X(e_1 , . . . , e_k )$} 

By Definition~\ref{def:application}, we have
\begin{align*}
 (t;\s) \in \tvrs{X(e_1,...,e_k)} 
    &~~\lequiv~~ (t;\s) \in \rho(X)@(\foint{e_1}_{\rho},...,\foint{e_k}_{\rho})\\
    &~~\lequiv~~ t\in (\rho(X))_\s(\foint{e_1}_{\rho}(\s),\ldots,\foint{e_k}_{\rho}(\s))\\
    &~~\lequiv~~ t\in \tvrs{X((e_1(\s))^\inj, \ldots, (e_k(\s))^\inj )} ~~\lequiv~~ t\in \tvrs{\trunc{X(e_1,...,e_k)}}.
\end{align*}

\prfcase{$A\to B$}
We have
\begin{align*}
(t;\s) \in \tvr{A \to B} 
&~~\Leftrightarrow~~ \forall (u;\s) \in \tvr{A}. (t\,u;\s) \in\tvr{B} \\
&~~\stackrel{\mathrm{(HI)}}{\Leftrightarrow}~~  \forall u \in \tvrs{\trunc{A}}. t\,u \in\tvrs{\trunc{B}}
~~\Leftrightarrow ~~t \in \tvrs{\trunc{A \to B}} .
\end{align*}

\prfcase{$\bvnat e B$} The proof is similar to the case $A \to B$.
\begin{align*}
(t;\s) \in \tvr{\bvnat e B} 
&~~\Leftrightarrow~~  (t\,\overline{n};\s) \in\tvr{B} \text{ where }  n=e(\s) \\
&~~\stackrel{\mathrm{(HI)}}{\Leftrightarrow}~~  
(t\,\overline{n};\s) \in\tvrs{\trunc B} \text{ where }  n=e(\s)
~~\Leftrightarrow ~~t \in \tvrs{\trunc{\bvnat e B}}. 
\end{align*}

\prfcase{$A_1\land A_2$}
We have
\begin{align*}
(t;\s) \in \tvr{A_1\land A_2}
&~~ \Leftrightarrow~~ (\pi_1(t);\s) \in \tvr{A_1} \land (\pi_2(t);\s)\in\tvr{ A_2}\\
&~~\stackrel{\mathrm{(HI)}}{\Leftrightarrow}~~  (\pi_1(t);\s) \in \tvr{\trunc{A_1}} \land (\pi_2(t);\s)\in\tvr{ \trunc{A_2}}\\
&~~ \Leftrightarrow~~ (t;\s) \in \tvr{\trunc{A_1\land A_2}}.
\end{align*}

\prfcase{$\forall x.A$}
We have
\begin{equation*}
(t;\s)\in \tvr{\forall x.A} 
~~\Leftrightarrow ~~\forall f \in \N^\S. (t;\s) \in \tvr{A}
~~\stackrel{\mathrm{(HI)}}{\Leftrightarrow}~~
\forall f \in \N^\S.t \in \tvrs{\trunc{A}}
~~\Leftrightarrow ~~
t \in \tvrs{\trunc{\forall x.A}}.    
\end{equation*}

\prfcase{$\exists x.A$}
We have
\begin{equation*}
(t;\s)\in \tvr{\exists x.A} 
~~\Leftrightarrow~~ \exists f \in \N^\S .(t;\s) \in \tvr{A}
~~\stackrel{\mathrm{(HI)}}{\Leftrightarrow}~~ \exists f \in \N^\S .t\in \tvrs{\trunc{A}}
~~\Leftrightarrow~~ t \in \tvrs{\trunc{\exists x.A}}.
\end{equation*}

The cases of the second-order quantifiers are similar to the corresponding first-order quantifiers. \qedhere
}

%% file: elimination.tex
This essentially follows from the glueing theorem and \autoref{def:application}.
Indeed, recall that by definition we have 
$\tvr{\forall X.A} = \bigcap_{F:\N^k\to\sat} \tvS{A}_{\rho,X\mapsto F}$.
Let us define the following function from $\N^k$ to $\sat$: 
\[ \F: (n_1,...,n_k) \mapsto  \tvr{B[x_1:=n^\inj_1,...,x_k:=n^\inj_k]} \]

We can prove by an easy induction on $A$
that $\tvS{A}_{\rho,X\mapsto \F}=\tvS{A[X(x_1,...,x_k):=B]}_{\rho}$,
from which the proposition follows trivially.
The only interesting case  is when $A\equiv X(x_1,...,x_n)$.
Let us write $f_1,...,f_k$ for $\foint{\rho(x_1)},...,\foint{\rho(x_k)}$.
We have:
\begin{align*}
\tvS{X(x_1,...,x_n)}_{\rho,X\mapsto \F}
&=~\F@(f_1,\ldots,f_k)\\
&=~\{(t;\s):t\in \F_\s(f_1(\s),\ldots,f_k(\s))\} \tag*{(by Def. \ref{def:application})}\\
&=~  \bigcup_{\s\in\S} \F_\s(f_1(\s),\ldots,f_k(\s))\times\{\s\} \\
&=~ \bigcup_{\s\in\S} \tvrs{\trunc{B[x_1:=f_1,...,x_k:=f_k}}\times\{\s\} \tag*{(by def. of $\F$)}\\
&=~ \tvr{B[x_1:=f_1,...,x_k:=f_k]} ~=~ \tvr{B} \tag*{(by Prop. \ref{r:glueing})}
\end{align*}

%% file: delta.tex
\begin{enumerate}
 \item By definition, $\tvS{\st{\delta}\mapsto \bot}=\Lambda\times\S$, which entails the result.
 \item Obvious from part 1.
 \item Follows from the fact that $\delta(\s)=\s$
 and that by part~\ref{r:T} of Lemma~\ref{r:lem_T}, for any $t$.
 \[(\lambda x.T\,x\,\get)\,t~~\reds{\s}{\s}~~ T\,t\,\get ~~\reds \s \s~~ t\,{\overline{\s}}\]
 \item The proof is similar to the proof of \autoref{r:relativized}. 
  Let $\X\in\sat$ be a predicate and $u$ be a term such that $(u;\s)\in\tvr{\fabv x.\neg \st x \to \X}$.
 In particular, the latter implies that for any term $t$, it holds that $(u\,\s\,t;\s)\in\X$.
 Since $\X$ is saturated, the result then follows from the fact that
 $T\,u\,\get\,t
 \reds \s \s ~ u\,\overline{\s}\,t$ 
 which is a consequence of part~\ref{r:T} of Lemma~\ref{r:lem_T}.\qedhere
 \end{enumerate}

%% file: diag.tex
Let $\s$ be an arbitrary state. 
 Following the proof of part~\ref{r:T_cor} of \autoref{r:lem_T},
 it is clearly enough to prove that 
 $(\lambda xy.\set\,y\,\anyt;\s) \real \fabvs{y}. y\leq \delta$
 (the rest of the proof is exactly the same replacing $\neg \st \delta$ with $\fabvs{y}. y\leq \delta$).
 Let $n\in\N$ and $t$ an arbitrary term.
Then
 \[
 (\lambda xy.\set\,y\,t)\,t\,\overline{n} ~\reds \s {\s}~
 \set\,\overline{n}\;t ~\reds {\s}{\s'}~ t 
 \]
 where $\s'=\max(n,\s)$. In particular, $n\leq \delta(\s')$ holds, hence $(t;\s')\in\tvS{n\leq \delta}$
 and we can conclude by anti-reduction.

%% file: llpo.tex
\newcommand{\realizer}[4]{ #1 & \ureal & #3& \textrm{\footnotesize (\autoref{#4})}\\[0.5em]}
\newcommand{\realizerabrv}[5]{ #1 & \ureal & #3& \textrm{\footnotesize (#5~\ref{#4})}\\[0.5em]}
\newcommand{\realizerh}[4]{\realizer{ #1}{#2}{{H_{A,B}\to} #3}{#4}}
\newcommand{\realizerabrvh}[5]{\realizerabrv{ #1}{#2}{{H_{A,B}\to} #3}{#4}{#5}}
\newcommand{\hcolor}[1]{#1}
\newcommand{\optH}[1]{}
\newcommand{\vsep}{&&&\\[-0.5em]}
$\begin{array}{r@{~}c@{~}l@{\quad}r}

\realizerabrv{  t_{\mathrm{LLPO}} }{ \hcolor{\lambda h.} \mw{(t_{\mathrm{aux}}\,\hcolor h\,\anyt\,\get)}{\inl{t_\delta\,x_1}}{\inl{t_\delta\,x_2}}}        
{ \textrm{LLPO}^{\mathrm{st}} \equiv \hcolor{\underbrace{\fabvs{x}.\fabvs{y}. (A(x) \vee B(y))}_{H_{A,B}}} \to (\fabvs{x}.A(x) \lor \fabvs{y}.B(y))}
{r:llpo}
{Thm.}

\vsep
\realizerabrv{t_\delta}{\lambda xy.x\,y\,\inr \loopp}
{(\fabv{y}.y\leq \delta \imp A(y) ) \to \fabvs{y}.A(y)}
{r:t_delta}
{Lem.}

\realizerabrvh{  t_{\mathrm{aux}} }{ \hcolor{\lambda h.}\rec\, u_0\, u_S}  
{ \optH{\hcolor{H_{A,B}} \to} \fabvs{x}.[\underbrace{(\fabv{z}.z\leq x \to A(z) )}_{A_{\leq x}}
\lor\underbrace{( \fabv{z}.z\leq x \to B(z)\big)}_{B_{\leq x}}] }
{r:t_aux}
{Cor.}

\vsep
\realizerabrvh{ u_0 }{ \mw{\hcolor{h}\,\anyt\,0\,\anyt\,0}{\inl{t_{\leq0}\,x_1}}{\inl{t_{\leq0}\,x_2}}}
{A_{\leq 0}\lor B_{\leq 0}}
{r:u0}
{Lem.}

\realizerabrvh{ u_\suc }{ \lambda wxd.\mww{d}{d}{t_\Delta\,\hcolor{h}\,d_1}{t_{\Delta}\,(t_\lor\,\hcolor{h})\,d_2}}
{\fabvs{x}.\left(A_{\leq x}\lor B_{\leq x} \to A_{\leq S(x)}\lor B_{\leq S(x)}\right)}
{r:usuc}
{Lem.}

\realizerabrvh{  t_{\Delta} }{\lambda \hcolor{h}wxv.  \mww{c}{\left(t_\Phi\,\hcolor{h}\,\,\anyt\,(\suc x)\,\anyt\,(\suc x)\right)}{\inl{t_{\leq\suc}\,v\,c_1}}{\inr{c_2}}}
{\optH{\hcolor{H_{A,B}} \to} \fabvs{x}.(A_{\leq x} \to A_{\leq S(x)}\lor B_{\leq S(x)})}
{r:t_Delta}
{Lem.}

\vsep
  
\realizerabrvh{t_\Phi }{\lambda \hcolor{h}wx.\rec\,(t_0\,\hcolor{h}\,w\,x)\,(t_\suc \,\hcolor{h}\,w\,x)}
{\optH{\hcolor{H_{A,B}} \to}
\fabvs{x}\fabvs{y}.(\underbrace{A(x) \lor B_{\leq y}}_{\Phi_{A,B}(x,y)})}
{r:t_phi}
{Cor.}

\realizerabrvh{t_0 }{\lambda \hcolor{h}wx.\mw {(\hcolor{h}\,\anyt\,x\,\anyt\,0)}{\inl{x_1}}{t_{\leq0}\,x_2}}
{\optH{\hcolor{H_{A,B}} \to}\fabvs{x}.\Phi_{A,B}(x,0)}
{r:t_0}
{Lem.}

\realizerabrvh{ t_\suc }{ \lambda \hcolor{h}wx zyv.\mww{v} {v}{\inl{v_1}}{\\&&\qquad
 \mww{h}{(\hcolor{h}\,\anyt\,x\,\anyt\,(\suc y))}{\inl{h_1}}{
 \inr{t_{\leq\suc}\,v_2\,h_2}}}} 
{\optH{\hcolor{H_{A,B}} \to}\fabvs{x}. \fabvs{y}.\left(\Phi_{A,B}(x,y) \to \Phi_{A,B}(x,S(y))\right)}
{r:t_suc}{Lem.}

\vsep

\realizerabrv{ t_{\leq0} }{\lambda xny.\mw{y}{x}{x_2} }
{ A(0) \to A_{\leq 0}}{r:t_leq0}{Lem.}

\realizerabrv{ t_{\leq\suc} }{ missing}
{\forall x.\left(A_{\leq x} \imp A(S(x)) \imp A_{\leq S(x)}\right)}
{r:t_leqsuc}{Lem.}

\realizerabrv{ t_{\lor} }
{ \lambda hzxwy.\mww{h}{(h\,z\,x\,w\,y)}{\inr{h_1}}{\inl{h_2}}}
{H_{A,B}\to H_{B,A}}
{r:t_lor}{Lem.}

\multicolumn{4}{l}{\small\textit{where  $A$  and  $B$ are internal formulas.}}\\
 \end{array}
$

%% file: los.tex
The proof goes by induction on the structure of $A$.
In the cases $\bvnat {e} A$, $X(e_1,\dots,e_n)$, $A\imp B$ 
and $A\land B$,
the result follows directly from the definitions.
The proof for the quantifiers is similar to the usual proof of \Los' theorem.

\prfcase{$\exists x.A$}
By the induction hypothesis, we have that for any $f\in\N^\S$,
\[ \tvq{A}_{\rho,x\mapsto f} = \{t:\{\s\in\S: t;\s \in\tvS{A}_{\rho,x\mapsto f}\}\in\U\}. \]
By glueing, we have that
$\tvS{A}_{\rho,x\mapsto f} = \tv{\trunc{A}}^\s_{\rho,x\mapsto f} = 
    \tv{\trunc{A}}^\s_{\rho,x\mapsto (f(\s))^\inj}$.
We want to prove that for any $t\in \Lambda $
\[
\exists f\in\N^\S. t \in \tvq{A}_{\rho,x\mapsto f}
\text{    iff   }
\{\s\in\S: t;\s \in\tvS{\exists x.A}_{\rho}\}\in\U. \]
Observe that, by glueing, the right-hand side is equivalent to $$
\{\s\in\S: \exists n\in \N. t \in\tv{A}^\s_{\rho,x\mapsto n^\inj}\}\in\U.$$

\rcase
Assume that there exists $f\in\N^\S$ such that $t \in \tvq{A}_{\rho,x\mapsto f}$.
Then it is easy to see that
\[ \{\s\in\S: t \in \tv{A}^\s_{\rho,x\mapsto (f(\s))^\inj}\} \subseteq \{\s\in\S: \exists n\in \N. t \in \tv{A}^\s_{\rho,x\mapsto n^\inj}\}, \]
and hence the result follows from the upwards closure of the ultrafilter.

\lcase 
Assume now that $E\defeq \{\s\in\S: \exists n\in \N. t \in\tv{A}^\s_{\rho,x\mapsto n^\inj}\}\in\U.$

For any $\s\in E$, using countable choice we can pick an integer $n_\s$ such that 
$t \in\tv{A}^\s_{\rho,x\mapsto n_\s^\inj}$.
We may then define the function $g\in\N^\S$ by:
\[g(\s) \defeq \left\{\begin{array}{rl}
   n_\s & \text{if } \s\in E\\
   0 & \text{otherwise}\\
   \end{array}\right.\]
By definition, $E\subseteq \{\s\in\S: t\in\tv{A}^\s_{\rho,x\mapsto (g(\s))^\inj}\}$,
hence this set belongs to $\U$ by upwards closure.
By induction hypothesis we conclude that
$t\in \tvq{A}_{\rho,x\mapsto f}$.

{
\prfcase{$\forall x.A$}
By the induction hypothesis, for any $f\in\N^\S$,
\[ \tvq{A}_{\rho,x\mapsto f} = \{t:\{\s\in\S: t;\s \in\tvS{A}_{\rho,x\mapsto f}\}\in\U\}. \]
By glueing,  
$\tvS{A}_{\rho,x\mapsto f} = \tv{\trunc{A}}^\s_{\rho,x\mapsto f} = 
    \tv{\trunc{A}}^\s_{\rho,x\mapsto (f(\s))^\inj}$.
We want to prove that for any $t\in \Lambda $
\[
\forall f\in\N^\S. t \in \tvq{A}_{\rho,x\mapsto f}
\text{    iff   }
\{\s\in\S: t;\s \in\tvS{\forall x.A}_{\rho}\}\in\U.\]
Observe that, by glueing, the right-hand side is equivalent to
\[
S\defeq\{\s\in\S: \forall n\in \N. t \in\tv{A}^\s_{\rho,x\mapsto n^\inj}\}\in\U.
\]

\rcase
We easily see that for any $f\in\N^\S$
\[ S=\{\s\in\S: \forall f\in\N^\S.t \tv{A}^\s_{\rho,x\mapsto (f(\s))^\inj}\} \subseteq \{\s\in\S: t \in \tv{A}^\s_{\rho,x\mapsto (f(\s))^\inj}\} \]
and by upwards closure we conclude that $t \in \tvq{A}_{\rho,x\mapsto f}$.

\lcase 
By contraposition, assume that 
$\{\s\in\S: \forall n\in \N. t \in\tv{A}^\s_{\rho,x\mapsto n^\inj}\}\notin\U$ and let us show that
there exists $f\in\N^\S$ such that $t \notin\tv{A}^\s_{\rho,x\mapsto f}$.
Because $\U$ is an ultrafilter, the assumption is equivalent to:
\[E=\overline{\{\s\in\S: \forall n\in \N. t \in\tv{A}^\s_{\rho,x\mapsto n^\inj}\}}
= \{\s\in\S: \exists n\in \N. t \notin\tv{A}^\s_{\rho,x\mapsto n^\inj}\}
\in \U.\]
We are essentially left with a situation similar to the existential case:
for any $\s\in E$, using countable choice we can pick an integer $n_\s$ such that 
$t \notin\tv{A}^\s_{\rho,x\mapsto n_\s^\inj}$.
We can then define the function $g\in\N^\S$  
such that for any $\s\in E$, $g(\s)= n_\s$. Hence
$E\subseteq \{\s\in\S: t\notin\tv{A}^\s_{\rho,x\mapsto (g(\s))^\inj}\}$,
and we conclude that $t\notin \tvq{A}_{\rho,x\mapsto f}$. \qedhere

}